\documentclass[a4paper,fleqn,11pt]{article}

\usepackage{amsmath}
\usepackage{amsthm}
\usepackage{amssymb}
\usepackage{authblk}

\usepackage[a4paper,top=3cm, bottom=3cm, left=3cm, right=3cm]{geometry}
\usepackage[shortlabels]{enumitem}
\usepackage[pdftex]{graphicx}
\usepackage{subcaption}
\usepackage[pdftex,ocgcolorlinks,pagebackref=false]{hyperref}

\theoremstyle{plain}
\newtheorem{theorem}{Theorem}[section]
\newtheorem{lemma}[theorem]{Lemma}
\newtheorem{proposition}[theorem]{Proposition}
\newtheorem{corollary}[theorem]{Corollary}
\newtheorem{remark}[theorem]{Remark}
\newtheorem{example}[theorem]{Example}

\theoremstyle{definition}
\newtheorem{definition}[theorem]{Definition}

\newcommand{\norm}[2][]{\left\|#2\right\|_{#1}}
\newcommand{\ball}[2]{B_{#1}(#2)}
\newcommand{\ket}[1]{\left|#1\right\rangle}

\newcommand{\ketbra}[2]{\left|#1\middle\rangle\!\middle\langle#2\right|}

\newcommand{\setbuild}[2]{\left\{#1\middle|#2\right\}}
\DeclareMathOperator{\domain}{dom}
\DeclareMathOperator{\relativeinterior}{relint}
\DeclareMathOperator{\boundeds}{\mathcal{B}}
\DeclareMathOperator{\states}{\mathcal{S}}
\DeclareMathOperator{\Tr}{Tr}
\DeclareMathOperator{\GL}{GL}
\DeclareMathOperator{\Hom}{Hom}
\DeclareMathOperator{\support}{supp}
\DeclareMathOperator{\principalminor}{pm}
\DeclareMathOperator{\id}{id}

\newcommand{\ed}{\mathop{}\!\mathrm{d}}

\let\Re\undefined

\DeclareMathOperator{\Re}{Re}

\newcommand{\ratefunction}[1]{I_{#1}}
\newcommand{\nonlinearpairing}[2]{{(#1,#2)_K}}
\newcommand{\pairing}[2]{{\langle #1,#2\rangle}}
\newcommand{\action}[2]{{#1\cdot #2}}
\newcommand{\positivechamber}{i\mathfrak{t}^*_+}
\DeclareMathOperator{\Ad}{Ad}

\title{Large deviation principle for moment map estimation}
\author[1]{Alonso Botero}
\author[2]{Matthias Christandl}
\author[3,4]{P\'eter Vrana}
\affil[1]{Departamento de F\'isica, Universidad de los Andes, Cra 1 No 18A-12, Bogot\'a, Colombia}
\affil[2]{QMATH, Department of Mathematical Sciences, University of Copenhagen, Universitetsparken 5, 2100 Copenhagen, Denmark}
\affil[3]{Institute of Mathematics, Budapest University of Technology and Economics, Egry J\'ozsef u. 1., 1111 Budapest, Hungary}
\affil[4]{MTA-BME Lend\"ulet Quantum Information Theory Research Group}

\date{}

\begin{document}
\maketitle
\begin{center}
\emph{We dedicate this work to the memory of Graeme Mitchison}
\end{center}

\begin{abstract}
Given a representation of a compact Lie group and a state we define a probability measure on the coadjoint orbits of the dominant weights by considering the decomposition into irreducible components. For large tensor powers and independent copies of the state we show that the induced probability distributions converge to the value of the moment map. For faithful states we prove that the measures satisfy the large deviation principle with an explicitly given rate function.
\end{abstract}

\section{Introduction}

This paper is concerned with probability distributions related to decompositions of tensor power representations into irreducibles. More precisely, given a representation $\pi$ of a compact Lie group $K$ on a finite dimensional Hilbert space $\mathcal{H}$ as well as a positive operator $\rho$ on $\mathcal{H}$ with unit trace (a state), for every $n$ the values $\Tr P_\lambda\rho^{\otimes n}$ determine a probability distribution on the set of dominant weights $\lambda$ of $K$.

The asymptotic behaviour of such distributions has been studied by many authors with different motivations. In the context of random walks, it was shown in \cite{grabiner1993random} that counting the multiplicities of irreducible representations in certain tensor power representations is equivalent to enumerating the number of walks for ``reflectable'' walk types, conditioned on staying within a Weyl chamber. This class of random walks was introduced by Gessel and Zeilberger in \cite{gessel1992random} as a generalisation of the classical ballot problem \cite{andre1887solution,zeilberger1983andre} to finite reflection groups. In \cite{tate2004lattice}, Tate and Zelditch analysed the multiplicities in high tensor powers and proved a central limit theorem as well as large deviation results by relating them to the (much simpler) weight multiplicity asymptotics. Moving away from the tracial state to positive operators arising from the representation of the complexified group, Postnova and Reshetikhin \cite{postnova2020multiplicities} generalised these asymptotic formulas to character distributions.

A different viewpoint is provided by interpreting these processes as restrictions of a random walk on a noncommutative space, the dual of the compact Lie group $K$, to a classical subalgebra \cite{biane1991quantum,parthasarathy1990generalized}. More precisely, given a state on the group von Neumann algebra $\mathop{vN}(K)$ one constructs a quantum Markov chain on the infinite tensor product $\mathop{vN}(K)^{\otimes\infty}$, and the probability distributions in question are obtained by restriction to the center $Z(\mathop{vN}(K))$, which can be identified with the $\ell^\infty$ space on the set of isomorphism classes of irreducible representations of $K$ (see also \cite{biane1995permutation}).

In the context of statistical mechanics, Ceg\l a, Lewis and Raggio investigated the multiplicities arising from the isotypic decomposition of tensor products of representations of the group $SU(2)$, and proved a large deviation principle \cite{cegla1988free}. Duffield \cite{duffield1990large} extended their result to an arbitrary compact semisimple Lie group using the G\"artner--Ellis theorem \cite[Theorem II.6.1.]{ellis1985entropy}.

In quantum statistics, Alicki, Rudicki and Sadowski \cite{alicki1988symmetry}, and later Keyl and Werner \cite{keyl2001estimating} proposed an estimator for the spectrum of the density operator, which is based on the decomposition of the tensor powers of the defining representation of $SU(d)$ (or $U(d)$). Based on Duffield's result, Keyl and Werner found the rate function for the exponential decay to be the relative entropy between the normalised Young diagram labelling the irreducible representation and the nonincreasingly ordered spectrum of the state. In \cite{keyl2006quantum} Keyl refined the estimator to a continuous positive operator valued measure (POVM) estimating both the spectrum and the eigenvectors of an unknown state and proved a large deviation principle in that setting. The appearance of the relative entropy in the result of Keyl and Werner suggests that similar large deviation rate functions should be viewed as information quantities. In \cite{vrana2023family} a family of entanglement measures have been constructed based on the rate function corresponding to the standard representation of products of unitary groups. While studying a tripartite extension of the Matsumoto--Hayashi universal distortion-free entanglement concentration protocol \cite{matsumoto2007universal}, Botero and Mej\'ia \cite{botero2018universal} found formulas for the probabilities induced by the isotypic decomposition or tensor powers of the $U(2)^n$-representation ${\mathbb{C}^2}^{\otimes n}$ when the state is in the $W$ class.

We consider the following common generalisation of these problems, including Keyl's refinement \cite{keyl2006quantum}. Let $K$ be a compact connected Lie group and $\pi:K\to U(\mathcal{H})$ a finite dimensional unitary representation. Then $\frac{1}{2\pi i}J:\states(\mathcal{H})\to\mathfrak{k}^*$ where $\pairing{J(\rho)}{\xi}=\Tr(\rho (T_e\pi)(\xi))$ can be regarded as an equivariant moment map ($T_e\pi$ is the derivative of $\pi$ at the identity). We wish to construct a measurement on $\mathcal{H}^{\otimes m}$ with outcomes in $i\mathfrak{k}^*$ that estimates $J(\rho)$ when applied to $m$ independent copies of a state $\rho$. The tensor powers of $\mathcal{H}$ can be decomposed into isotypic components as
\begin{equation}
\mathcal{H}^{\otimes m}\simeq\bigoplus_{\lambda}\mathcal{H}_\lambda\otimes\Hom_K(\mathcal{H}_\lambda,\mathcal{H}^{\otimes m}),
\end{equation}
where the sum is over dominant integral weights $\lambda$, $\mathcal{H}_\lambda$ is an irreducible representation of $K$ with highest weight $\lambda$ and $\Hom_K(\mathcal{H}_\lambda,\mathcal{H}^{\otimes m})$ is the multiplicity space. Let $\ket{v_\lambda}$ be a highest weight vector of norm $1$. Then $\action{K}{\ketbra{v_\lambda}{v_\lambda}}$ can be identified (via the moment map) with the orbit of $\lambda$ in $i\mathfrak{k}^*$ under the (complexified) coadjoint action. Weighted with the suitably normalised invariant measure on the orbit and tensored with the identity operator on the multiplicity space, these projections give rise to a POVM $E_{\mathcal{H}^{\otimes m}}$ from the Borel $\sigma$-algebra of $i\mathfrak{k}^*$ to $\boundeds(\mathcal{H}^{\otimes m})$. In the special case $K=U(d)$ and $\pi$ the standard representation, this measure is the same as the one proposed in \cite{keyl2006quantum}.

If $\rho$ is a state on $\mathcal{H}$ then we can form the sequence of probability measures $\mu_m(A)=\Tr\rho^{\otimes m}E_{\mathcal{H}^{\otimes m}}(mA)$. These measures are interpreted as the probability distribution of the (rescaled) random classical outcome of the measurements that are described by the POVM $E_{\mathcal{H}^{\otimes m}}$. We find that the measures $\mu_m$ converge weakly to the Dirac measure concentrated at $J(\rho)$ and the convergence is exponentially fast.

To formulate a more precise statement, choose a Borel subgroup $B$ in the complexification of $K$, let $N$ be its maximal unipotent subgroup, let $\mathfrak{a}=i\mathfrak{t}$ where $\mathfrak{t}$ is the Lie algebra of the maximal torus $T=K\cap B$, and let $\positivechamber\subseteq i\mathfrak{t}^*$ be the closure of the positive Weyl chamber. Then every element $x\in i\mathfrak{k}^*$ can be written as $x=\action{h}{x_0}$ with a unique $x_0\in \positivechamber$ and a non-unique $h\in K$. Let
\begin{equation}
\ratefunction{\rho}(\action{h}{x_0})=\sup_{\alpha\in\mathfrak{a}}\max_{n\in N}\pairing{x_0}{\alpha}-\ln\Tr\pi(n)^*\pi(\exp\alpha/2)\pi(h)^*\rho\pi(h)\pi(\exp\alpha/2)\pi(n).
\end{equation}
We will show in Section~\ref{sec:ratefunction} that $\ratefunction{\rho}$ is well defined and is a good rate function (i.e. has compact sub-level sets).
\begin{theorem}[Large deviation principle]\label{thm:LDP}
Let $\mathcal{H}$ be a finite dimensional Hilbert space, $K$ a compact connected group, $\pi:K\to U(\mathcal{H})$ and $\mu_m$ as above.
\begin{enumerate}[(i)]
\item\label{it:LDPupper} For every state $\rho$ and closed subset $C\subseteq i\mathfrak{k}^*$ we have
\begin{equation}
\limsup_{m\to\infty}\frac{1}{m}\ln\mu_m(C)\le-\inf_{x\in C}\ratefunction{\rho}(x).
\end{equation}
\item\label{it:LDPlower} For every faithful state $\rho$ and open subset $O\subseteq i\mathfrak{k}^*$ we have
\begin{equation}\label{eq:LDPlower}
\liminf_{m\to\infty}\frac{1}{m}\ln\mu_m(O)\ge-\inf_{x\in O}\ratefunction{\rho}(x).
\end{equation}
\end{enumerate}
\end{theorem}
For faithful states the theorem says that the measures $\mu_m$ satisfy the large deviation principle with rate function $\ratefunction{\rho}$. For general states we can only prove a weaker version of \eqref{eq:LDPlower}, replacing $O$ on the right hand side with $O\cap\mathcal{M}_\rho$ where $\mathcal{M}_\rho$ is a dense subset of $\domain\ratefunction{\rho}$. In the examples below $\ratefunction{\rho}$ is continuous on its domain, therefore the stronger conclusion \eqref{eq:LDPlower} still holds even if $\rho$ is not invertible.

In addition, we show that $\ratefunction{\rho}$ only vanishes at $J(\rho)$, which identifies the weak limit as a Dirac measure:
\begin{theorem}[Law of large numbers]\label{thm:LLN}
For every state $\rho$ and open set $O\subseteq i\mathfrak{k}^*$ such that $J(\rho)\in O$ we have
\begin{equation}
\lim_{m\to\infty}\mu_m(O)=1.
\end{equation}
\end{theorem}

If one is only interested in the way $\rho^{\otimes m}$ distributes the probabilty among the isotypic components, then it is possible to extract this coarse-grained information by taking the pushforward of the measure $\mu_m$ along the continuous function that sends $x\in i\mathfrak{k}^*$ to the unique element $x_0\in (\action{K}{x})\cap \positivechamber$. According to the contraction principle, for invertible states these measures also satisfy the large deviation principle with rate function
\begin{equation}\label{eq:contraction}
\begin{split}
\tilde{I}_\rho(x_0)
 & = \inf_{h\in K}\ratefunction{\rho}(\action{h}{x_0})  \\
 & = \inf_{h\in K}\sup_{\alpha\in\mathfrak{a}}\max_{n\in N}\pairing{x_0}{\alpha}-\ln\Tr\pi(n)^*\pi(\exp\alpha/2)\pi(h)^*\rho\pi(h)\pi(\exp\alpha/2)\pi(n).
\end{split}
\end{equation}

Our result reduces to the formula of Ceg\l a, Lewis and Raggio \cite{cegla1988free} and to \cite[Theorem 2.1.]{duffield1990large} and \cite[Corollary 4.]{tate2004lattice} when $\rho=\frac{I}{\dim\mathcal{H}}$ and to Keyl's rate function \cite[Theorem 3.2]{keyl2006quantum} when $K=U(d)$, $\mathcal{H}=\mathbb{C}^d$ and $\pi$ is the standard representation. From the latter the result of Keyl--Werner \cite{keyl2001estimating} follows by the contraction principle. Setting $K=U(1)^d$ we also recover Cram\'er's large deviation theorem \cite{cramer1938nouveau} in the special case of finitely supported integer-valued random variables.

The key element of our proof can be viewed as a noncommutative generalisation of the exponential tilting method, applied directly to the quantum states \emph{before} projecting onto the classical subalgebra (or pairing with the POVM). More precisely, we replace the state $\rho^{\otimes n}$ with the transformed state $(\pi(g)^*\rho\pi(g))^{\otimes n}$, where $g$ is an element of the complexification of $K$. Consequently, the transformed measures are not obtained by multiplication with a suitable function and must be related to the original one more carefully, also taking into account the (complexified) coadjoint action of $K$. A substantial part of our work is the development of these techniques in Section~\ref{sec:coadjoint}.

The paper is organised as follows. In Section~\ref{sec:preliminaries} we fix the notation and collect some facts related to the representations and structure of compact Lie groups and their complexifications, and to positive operator valued measures. Section~\ref{sec:estimation} contains the proof of our main results: in Section~\ref{sec:measurements} we define the POVM used in the estimation scheme, in Section~\ref{sec:coadjoint} we introduce an action of the complexification of $K$ on $i\mathfrak{k}^*$, in Section~\ref{sec:ratefunction} we give several formulas for the rate function and prove some of its key properties, in Section~\ref{sec:upperbound} we prove the large deviation upper bound, in Section~\ref{sec:lawoflargenumbers} we prove weak convergence to the value of the moment map and in Section~\ref{sec:lowerbound} we address the large deviation lower bound.

\paragraph{Related work.}
Closely related work has been done independently by Cole Franks and Michael Walter \cite{franks2020minimal}.

\section{Preliminaries}\label{sec:preliminaries}

Throughout Hilbert spaces are assumed to be finite dimensional and the inner product $\langle\cdot,\cdot\rangle$ is linear in the second argument. $U(\mathcal{H})$ denotes the group of unitary operators on $\mathcal{H}$. We denote by $\states(\mathcal{H})$ the set of positive semidefinite operators on $\mathcal{H}$ with trace equal to $1$ (states). Every state $\rho\in\states(\mathcal{H})$ admits a purification, i.e. $\rho=\Tr_{\mathbb{C}^d}\ketbra{\psi}{\psi}$ for some unit vector $\psi\in\mathcal{H}\otimes\mathbb{C}^d$, where $\Tr_{\mathbb{C}^d}$ is the partial trace and $\ketbra{\psi}{\psi}$ is the orthogonal projection onto the subspace spanned by $\psi$. Below we will collect the necessary standard results on compact Lie groups and their complexifications. More details can be found in many textbooks, e.g. \cite{knapp2001representation}.

\subsection{Complexification}

Let $K$ be a compact connected Lie group, $\mathfrak{k}=T_eK$ its Lie algebra ($T_e$ stands for the tangent space at the identity when applied to a Lie group and the derivative at the identity when applied to homomorphisms). The complexification $G=K_\mathbb{C}$ is a complex Lie group together with an inclusion $K\to G$, defined by the property that every smooth homomorphism from $K$ to a complex Lie group extends uniquely to a holomorphic homomorphism from $G$. $G$ is a reductive group and its Lie algebra is $\mathfrak{g}=\mathbb{C}\otimes_{\mathbb{R}}\mathfrak{k}$. We identify $\mathfrak{k}^*$ with the subspace of $\mathfrak{g}^*$ that consists of functionals that are real on $\mathfrak{k}$. We use angle brackets $\pairing{\cdot}{\cdot}$ for the pairing between a vector space and its dual.

The group multiplication gives a diffeomorphism $K\times P\to G$ where $P=\exp(i\mathfrak{k})$. The (global) Cartan involution is a group homomorphism $\Theta:G\to G$ that fixes $K$ and acts on $p\in P$ as $\Theta(p)=p^{-1}$. $K$ is the fixed point set of $\Theta$. For $g\in G$ we define $g^*=\Theta(g)^{-1}$. The exponential map provides a diffeomorphism between $i\mathfrak{k}$ and $P$.

A finite dimensional unitary representation $\pi:K\to U(\mathcal{H})$ extends uniquely to a homomorphism $G\to\GL(\mathcal{H})$ as complex Lie groups. The extension, denoted with the same symbol, satisfies $\pi(g^*)=\pi(g)^*$.

\begin{example}\label{ex:complexification:torus}
Let $K=U(1)^d$ be a torus. Then $\mathfrak{k}$ can be identified with $i\mathbb{R}^d$, $G=(\mathbb{C}^\times)^d$ and $\mathfrak{g}$ is the Lie algebra $\mathbb{C}^d$. $i\mathfrak{k}\simeq\mathbb{R}^d$ and $P\simeq\mathbb{R}_{>0}^d$. $\Theta$ sends a $d$-tuple $g=(g_1,\ldots,g_d)\in G$ to $(\overline{g_1}^{-1},\ldots,\overline{g_d}^{-1})$ and $g^*$ is the (componentwise) conjugate of $g$.
\end{example}

\begin{example}\label{ex:complexification:unitary}
Let $K=U(d)$ be the group of $d\times d$ unitary matrices. Then $\mathfrak{k}$ consists of skew-hermitian matrices, $G=\GL(d,\mathbb{C})$ and $\mathfrak{g}$ is the Lie algebra of all $d\times d$ complex matrices. We identify $\mathfrak{g}^*$ with $\mathfrak{g}$ via the pairing $(x,\xi)\mapsto\Tr(x\xi)$. Under this identification $\mathfrak{k}^*$ corresponds to the space of skew-hermitian matrices.

$i\mathfrak{k}$ is the space of hermitian matrices and $P$ is the set of positive definite matrices. $\Theta$ sends a matrix to the conjugate transpose of its inverse, and $g^*$ is the conjugate transpose of $g$.
\end{example}

\subsection{Moment map}\label{sec:momentmap}

Let $\pi:K\to U(\mathcal{H})$ be a representation on a finite dimensional Hilbert space. We consider the map $J:\states(\mathcal{H})\to i\mathfrak{k}^*$ defined as
\begin{equation}
\pairing{J(\rho)}{\xi}=\Tr\left(T_e\pi(\xi)\rho\right).
\end{equation}
We identify the projective space $P\mathcal{H}$ with the set of rank one orthogonal projections. For $v\in\mathcal{H}\setminus\{0\}$ we denote the corresponding projection (or equivalence class) by $[v]$. In particular, the restriction of $J$ to $P\mathcal{H}$ is
\begin{equation}
\pairing{J([v])}{\xi}=\frac{\langle v,T_e\pi(\xi)v\rangle}{\norm{v}^2}.
\end{equation}
$P\mathcal{H}$ is a symplectic manifold with the Fubini--Study symplectic form and the action of $K$ is Hamiltonian with $\frac{1}{2\pi i}J$ as a moment map \cite[2.7]{kirwan1984cohomology}.

In the special case when $\pi:K\to U(\mathcal{H}_\lambda)$ is an irreducible representation with highest weight $\lambda$ we will use the notation $J_{\lambda}$ for the above map. Its value on the ray through a highest weight vector $v_\lambda$ is $J_{\lambda}([v_\lambda])=\lambda$. The restriction of $J_{\lambda}$ to $\action{K}{[v_\lambda]}$ is injective.

\begin{remark}
Let $\rho\in\states(\mathcal{H})$ and $\psi\in\mathcal{H}\otimes\mathbb{C}^d$ be any purification of $\rho$. Consider the representation $K\to U(\mathcal{H}\otimes\mathbb{C}^d)$ given by $k\mapsto\pi(k)\otimes I$. Then
\begin{equation}
\pairing{J_{\mathcal{H}\otimes\mathbb{C}^d}([\psi])}{\xi}=\langle \psi,\left(T_e\pi(\xi)\otimes I\right)\psi\rangle=\Tr\left(T_e\pi(\xi)\rho\right)=\pairing{J(\rho)}{\xi}.
\end{equation}
Thus in general $\frac{1}{2\pi i}J\circ\Tr_{\mathbb{C}^d}$ is a moment map for $K$ acting on the larger projective space.
\end{remark}

$K$ acts on $\states(\mathcal{H})$ as $\action{k}{\rho}:=\pi(k)\rho\pi(k)^*$. For $k\in K$ we have
\begin{equation}
\begin{split}
\pairing{J(\action{k}{\rho})}{\xi}
 & = \Tr\left(T_e\pi(\xi)\pi(k)\rho\pi(k)^*\right)  \\
 & = \Tr\left(\pi(k^{-1})T_e\pi(\xi)\pi(k)\rho\right)  \\
 & = \Tr\left(T_e\pi(\action{k^{-1}}{\xi})\rho\right)  \\
 & = \pairing{J(\rho)}{\action{k^{-1}}{\xi}}  \\
 & = \pairing{\action{k}{J(\rho)}}{\xi},
\end{split}
\end{equation}
i.e. $J$ is equivariant with respect to the coadjoint action of $K$ on $i\mathfrak{k}^*$.
\begin{remark}
$P\mathcal{H}$ is connected, therefore any two moment maps differ by a constant, and for any two equivariant moment maps the difference is fixed by the coadjoint action. Thus when $K$ is semisimple, $\frac{1}{2\pi i}J$ is the unique equivariant moment map.
\end{remark}

\begin{example}\label{ex:momentmap:torus}
Let $K=U(1)^d$. Irreducible representations are one dimensional and are of the form $\pi_n(g_1,\ldots,g_d)v=g_1^{n_1}g_2^{n_2}\cdots g_d^{n_d}v$ for $n=(n_1,\ldots,n_d)\in\mathbb{Z}^d$. Let $\pi:K\to U(\mathcal{H})$ be an arbitrary representation and decompose $\mathcal{H}$ into isotypic subspaces as $\mathcal{H}=\bigoplus_{n\in\mathbb{Z}^d}\mathcal{H}_n\otimes\mathbb{C}^{d_n}$ ($d_n$ is nonzero for only finitely many terms). If $P_n$ denotes the orthogonal projection onto these subspaces, then $J(\rho)$ is determined by the numbers $r_n=\Tr P_n\rho$ as $\pairing{J(\rho)}{\xi}=\sum_{n\in\mathbb{Z}^d}r_n\sum_{i=1}^d n_i\xi_i$ for $\xi=(\xi_1,\ldots,\xi_d)\in\mathbb{C}^d\simeq\mathfrak{g}$.
\end{example}

\begin{example}\label{ex:momentmap:unitary}
Let $K=U(d)$ and $\pi$ the standard representation on $\mathbb{C}^d$. Under the identification of $i\mathfrak{k}^*$ with the space of hermitian matrices (see Example~\ref{ex:complexification:unitary}) $J$ maps every state $\rho$ to itself.
\end{example}

\subsection{Borel subgroups}

A Borel subgroup $B\le G$ is a maximal solvable subgroup. Any two such subgroups are conjugate by an element of $K$. From now on we fix a Borel subgroup $B$ and use the following notations: $T=B\cap K$, $\mathfrak{t}=T_eT$, $\mathfrak{a}=i\mathfrak{t}\subseteq\mathfrak{g}$, $A=\exp\mathfrak{a}$, $N=[B,B]$, $\mathfrak{b}=T_eB$, $\mathfrak{n}=T_eN$. $T$ is a maximal torus in $K$, $\exp:\mathfrak{a}\to A$ is a diffeomorphism, $N$ is the maximal unipotent subgroup of $B$ and $T$ normalises $N$. For any element $g\in G$ we have the Iwasawa decomposition $g=kan$ with uniquely determined elements $k\in K$, $a\in A$ and $n\in N$. For $g\in G$ we write $\alpha(g)$ for the element of $\mathfrak{a}$ that exponentiates to $a$. The map $\alpha:G\to\mathfrak{a}$ is smooth. In a similar way we write $k(g)$ for the $K$-component of $g$ in the Iwasawa decomposition.

\begin{example}\label{ex:Borel:torus}
Let $K=U(1)^d$. Then $B=G$, $T=K$, $A=P$ and $N=\{e\}$. The Iwasawa decomposition is the same as the componentwise polar decomposition, more precisely $g=(g_1,\ldots,g_d)=(g_1/|g_1|,\ldots,g_d/|g_d|)\cdot(|g_1|,\ldots,|g_d|)\cdot e$.
\end{example}

\begin{example}\label{ex:Borel:unitary}
Let $K=U(d)$. Then one choice for $B$ is the subgroup of upper triangular matrices, $T$ is the group of diagonal unitaries, $A$ is the group of positive definite diagonal matrices, $N$ consists of upper triangular matrices with $1$ on the main diagonal. The Iwasawa decomposition is essentially the QR decomposition, but with the triangular part decomposed into its diagonal and another upper triangular matrix with $1$ entries on the main diagonal.
\end{example}

If $\pi:K\to U(\mathcal{H})$ is a representation, then $v\in\mathcal{H}\setminus\{0\}$ is a highest weight vector for $B$ if $[v]=[\pi(b)v]$ for every $b\in B$ or equivalently, $v=\pi(n)v$ for every $n\in N$ and $v$ is an eigenvector of $A$. Let $g\in G$ and let $g=kan$ be its Iwasawa decomposition. Then $[\pi(g)v]=[\pi(k)\pi(an)v]=[\pi(k)v]$, therefore $g$ is in the stabiliser of $[v]$ iff $k=k(g)$ is in the stabiliser.

Let $\positivechamber\subseteq i\mathfrak{k}^*$ be the closure of the dominant Weyl chamber, using an $\Ad$-equivariant inner product to identify $i\mathfrak{t}^*$ with a subspace of $i\mathfrak{k}^*$. Then every coadjoint orbit in $i\mathfrak{k}^*$ intersects $\positivechamber$ in a unique point. For $x\in i\mathfrak{k}^*$ let $K_x\le K$ be the stabiliser subgroup with respect to the coadjoint action. Writing $x=\action{h}{x_0}$ with $x_0\in\positivechamber$ and $h\in K$, the subgroup of $G$ generated by $K_x$ and $hBh^{-1}$ is a parabolic subgroup and its intersection with $K$ is $K_x$.

\subsection{Positive operator valued measures}\label{sec:POVM}

Let $(X,\mathcal{X})$ be a measurable space and $\mathcal{H}$ a Hilbert space. A positive operator valued measure is a $\sigma$-additive map $E:\mathcal{X}\to\boundeds(\mathcal{H})$ such that $E(A)\ge0$ for every $A\in\mathcal{X}$, $E(X)=I$.

We will construct positive operator valued measures in the following way. Let $\mu_0$ be a measure on $(X,\mathcal{X})$ and let $f:X\to\boundeds(\mathcal{H})$ be a function that is measurable, its values are ($\mu_0$-almost everywhere) positive operators and
\begin{equation}
\int_X f\ed\mu_0=I.
\end{equation}
Then
\begin{equation}
E(A)=\int_A f\ed\mu_0
\end{equation}
defines a positive operator valued measure, which will be denoted $f\mu_0$.

Given a state $\rho\in\states(\mathcal{H})$ and a positive operator valued measure $E:\mathcal{X}\to\boundeds(\mathcal{H})$ we can form a probability measure $\mu:\mathcal{X}\to[0,1]$ as $\mu(A)=\Tr E(A)\rho$. In particular, when $E=f\mu_0$ with $\mu_0$ and $f$ as above, then
\begin{equation}
\mu(A)=\int_A\Tr(f(x)\rho)\ed\mu_0(x).
\end{equation}

\section{Moment map estimation}\label{sec:estimation}

In this section we define precisely our estimation scheme (Section~\ref{sec:measurements}) and prove our main results, a large deviation principle and a law of large numbers for the induced measures. The proofs are divided into separate sections: in Section~\ref{sec:coadjoint} we introduce an action of $G$ on $i\mathfrak{k}^*$ as well as a function (Definitions~\ref{defn:extendedaction} and~\ref{def:chi}) which encode the $G$-actions on the highest weight orbits of irreducible $K$-representations and extend them in a continuous and scale-equivariant way; in Section~\ref{sec:ratefunction} we present various expressions for the rate function and prove that it is a good rate function; Proposition~\ref{prop:LDPupper} in Section~\ref{sec:upperbound} proves part~\ref{it:LDPupper} of Theorem~\ref{thm:LDP}; Section~\ref{sec:lawoflargenumbers} contains the proof of Theorem~\ref{thm:LLN}; part~\ref{it:LDPlower} of Theorem~\ref{thm:LDP} is proved in Section~\ref{sec:lowerbound}.

From now on $K$ will be an arbitrary but fixed compact connected Lie group, $B$ a Borel subgroup of its complexification $G$, $\pi:K\to U(\mathcal{H})$ a finite dimensional representation and $\rho\in\states(\mathcal{H})$. In this generality we can prove a large deviation upper bound and a law of large numbers, whereas for the matching lower bound we will make the additional assumption $\support\rho=\mathcal{H}$.

\subsection{The measurements}\label{sec:measurements}

Let $\mathcal{H}_{\lambda}$ be an irreducible representation of $K$ with highest weight $\lambda$. The orbit of the highest weight ray in $P\mathcal{H}_\lambda$ has a unique $K$-invariant probability measure $\nu_\lambda$. Every representation $K\to U(\mathcal{K})$ can be decomposed as
\begin{equation}
\mathcal{K}=\bigoplus_{\lambda}\mathcal{H}_\lambda\otimes\Hom_K(\mathcal{H}_\lambda,\mathcal{K}),
\end{equation}
where the sum is over the integral weights in $\positivechamber$. We define $p_{\lambda,\mathcal{K}}:P\mathcal{H}_\lambda\to\boundeds(\mathcal{H}_\lambda\otimes\Hom_K(\mathcal{H}_\lambda,\mathcal{K}))\subseteq\boundeds(\mathcal{K})$ to be the function which sends the equivalence class of the unit vector $v$ to $\ketbra{v}{v}\otimes\id_{\Hom_K(\mathcal{H}_\lambda,\mathcal{K})}$. With the notation of Section~\ref{sec:POVM}, we take $X$ to be the orbit of the highest weight ray with its Borel $\sigma$-algebra as $\mathcal{X}$, the probability measure $\nu_\lambda$ plays the role of $\mu_0$ and $(\dim\mathcal{H}_\lambda)p_{\lambda,\mathcal{K}}$ is the measurable function $f$. Thus $(\dim\mathcal{H}_\lambda)p_{\lambda,\mathcal{K}}\nu_\lambda$ is a $\boundeds(\mathcal{H}_\lambda\otimes\Hom_K(\mathcal{H}_\lambda,\mathcal{K}))$-valued POVM.

Next we glue these together into a POVM on $i\mathfrak{k}^*$ by taking the pushforward and summing over the isomorphism classes of irreducible representations. Let $J_\lambda:P\mathcal{H}_\lambda\to i\mathfrak{k}^*$ be the map as in Section~\ref{sec:momentmap}. We define the positive operator valued measure
\begin{equation}\label{eq:measurementdef}
E_{\mathcal{K}}=\sum_{\lambda}(J_\lambda)_*((\dim\mathcal{H}_\lambda)p_{\lambda,\mathcal{K}}\nu_\lambda).
\end{equation}

For every $m\in\mathbb{N}$, $E_{\mathcal{H}^{\otimes m}}$ corresponds to a measurement that can be performed on $m$ copies of the state $\rho$, with values in $i\mathfrak{k}^*$. As we will see, the typical values behave in an extensive way. For this reason we will include a $\frac{1}{m}$ rescaling in the probability measures. Explicitly, for Borel sets $A\subseteq i\mathfrak{k}^*$ we define
\begin{equation}\label{eq:muintegral}
\begin{split}
\mu_m(A)
 & = \Tr\rho^{\otimes m}E_{\mathcal{H}^{\otimes m}}(mA)  \\
 & = \sum_{\lambda}\dim(\mathcal{H}_\lambda)\int_{J_\lambda^{-1}(mA)}\Tr\rho^{\otimes m}p_{\lambda,\mathcal{H}^{\otimes m}}([v_\lambda])\ed\nu_\lambda([v_\lambda]).
\end{split}
\end{equation}

\subsection{Extension of the coadjoint action}\label{sec:coadjoint}

The aim of this section is to define a continuous action of $G$ on $i\mathfrak{k}^*$ such that the restriction of $J_\lambda$ to the highest weight orbit is $G$-equivariant for every dominant weight $\lambda$, and the action commutes with scaling by nonnegative real numbers.

Consider an irreducible representation $\pi_\lambda:K\to U(\mathcal{H}_\lambda)$ and let $v_\lambda$ be a highest weight vector of norm $1$. The $K$-orbit of $[v_\lambda]$ is $G$-invariant (since $B$ fixes $[v_\lambda]$). While the action of $K$ keeps the norm fixed, elements of $G$ will change the norm. For $g\in G$ and an element $\ketbra{\action{h}{v_\lambda}}{\action{h}{v_\lambda}}$ in the orbit, let us write $gh=kan$ for the Iwasawa decomposition. We have
\begin{equation}\label{eq:GactionHWorbit}
\begin{split}
\pi_\lambda(g)\ketbra{\action{h}{v_\lambda}}{\action{h}{v_\lambda}}\pi_\lambda(g)^*
 & = \pi_\lambda(gh)\ketbra{v_\lambda}{v_\lambda}\pi_\lambda(gh)^*  \\
 & = \pi_\lambda(kan)\ketbra{v_\lambda}{v_\lambda}\pi_\lambda(kan)^*  \\
 & = \pi_\lambda(ka)\ketbra{v_\lambda}{v_\lambda}\pi_\lambda(ka)^*  \\
 & = e^{2\pairing{\lambda}{\alpha(gh)}}\pi_\lambda(k)\ketbra{v_\lambda}{v_\lambda}\pi_\lambda(k)^*  \\
 & = e^{2\pairing{\lambda}{\alpha(gh)}}\ketbra{\action{k(gh)}{v_\lambda}}{\action{k(gh)}{v_\lambda}},
\end{split}
\end{equation}
in the last line emphasizing the dependence $k=k(gh)$.

The restriction of $J_\lambda$ to $K\cdot[v_\lambda]$ is injective and $K$-equivariant, therefore we can use it to define a $G$-action on $K\cdot\lambda=J_\lambda(K\cdot[v_\lambda])$ and the constant factor in eq. \eqref{eq:GactionHWorbit} gives rise to a function $(g,\action{h}{\lambda})\mapsto 2\pairing{\lambda}{\alpha(gh)}$ on $\action{K}{\lambda}$.

For different representations these actions and functions are additive in the sense that $\action{k(gh)}{(\lambda_1+\lambda_2)}=\action{k(gh)}{\lambda_1}+\action{k(gh)}{\lambda_2}$ and $2\pairing{\lambda_1+\lambda_2}{\alpha(gh)}=2\pairing{\lambda_1}{\alpha(gh)}+2\pairing{\lambda_2}{\alpha(gh)}$. This is a consequence of the fact that $\pi_{\lambda_1}\otimes\pi_{\lambda_2}$ has an irreducible subrepresentation isomorphic to $\pi_{\lambda_1+\lambda_2}$ and $v_{\lambda_1}\otimes v_{\lambda_2}$ is a highest weight vector for this subrepresentation.

In what follows we extend both the $G$-actions and the functions giving the constant factor to $i\mathfrak{k}^*$ in a continuous way. The extensions will still be additive in the above sense (i.e. on $\action{h}{\positivechamber}$ for every $h\in K$) as well as positively homogeneous of degree $1$. Note that there is at most one such extension since the positive Weyl chamber is the cone generated by the dominant weights. For the action, the only possibility is that $g$ maps $\action{h}{x_0}$ to $\action{k(gh)}{x_0}$ when $h\in K$ and $x_0\in\positivechamber$. The next proposition shows that this is indeed well defined and determines an action of $G$ on $i\mathfrak{k}^*$
\begin{proposition}\label{prop:welldefinedaction}\leavevmode
\begin{enumerate}[(i)]
\item Let $x=\action{h_1}{x_0}=\action{h_2}{x_0}$ and $g\in G$. Then $\action{k(gh_1)}{x_0}=\action{k(gh_2)}{x_0}$.
\item\label{it:Gaction} Let $g_1,g_2\in G$, $h\in K$. Then $k(g_2g_1h)=k(g_2k(g_1h))$
\end{enumerate}
\end{proposition}
\begin{proof}\leavevmode
\begin{enumerate}[(i)]
\item Let $k_1a_1n_1=gh_1$ and $k_2a_2n_2=gh_2$ be the Iwasawa decompositions. Then
\begin{equation}
\begin{split}
k_2^{-1}k_1
 & = (gh_2n_2^{-1}a_2^{-1})^{-1}gh_1n_1^{-1}a_1^{-1}  \\
 & = a_2n_2h_2^{-1}g^{-1}gh_1n_1^{-1}a_1^{-1}  \\
 & = a_2n_2h_2^{-1}h_1n_1^{-1}a_1^{-1}.
\end{split}
\end{equation}
Since $h_2^{-1}h_1\in K_{x_0}$, the product is in the subgroup generated by $B$ and $K_{x_0}$. This is a parabolic subgroup whose intersection with $K$ is $K_{x_0}$, therefore $k_2^{-1}k_1\in K_{x_0}$, i.e. $\action{k_1}{x_0}=\action{k_2}{x_0}$.
\item Let $g_1h=k_1a_1n_1$ and $g_2k_1=k_2a_2n_2$ be the Iwasawa decompositions (so $k_1=k(g_1h)$ and $k_2=k(g_2k_1)=k(g_2k(g_1h))$). Then
\begin{equation}
g_2g_1h=k_2a_2n_2k_1^{-1}k_1a_1n_1=k_2a_2a_1(a_1^{-1}n_2a_1)n_1.
\end{equation}
The last two factors are in $N$ since $A$ normalises $N$, therefore $k_2=k(g_2g_1h)$.
\end{enumerate}
\end{proof}
This proposition justifies the following definition and also shows that it defines an action of $G$ on $i\mathfrak{k}^*$. Notice that we use the same notation for the newly introduced map as for the (complexified) coadjoint action of $K$. This will not lead to a confusion as the two agree on $K$.
\begin{definition}\label{defn:extendedaction}
We define a map $G\times i\mathfrak{k}^*\to i\mathfrak{k}^*$, $(g,x)\mapsto\action{g}{x}$ as follows. Let $x=\action{h}{x_0}$ with $x_0\in\positivechamber$ and $h\in K$ (coadjoint action). Then we set $\action{g}{x}:=\action{k(gh)}{x_0}$.
\end{definition}
By construction, the restriction of $J_{\lambda}$ to the highest weight orbit is $G$-equivariant for every dominant weight $\lambda$.

For $x\in i\mathfrak{k}^*$ we will denote by $G_x$ the stabiliser subgroup with respect to this action. If $x=\action{h}{x_0}$ with $x_0\in\positivechamber$ and $h\in K$ then $G_x$ is the (parabolic) subgroup generated by $K_x$ and $hBh^{-1}$ and satisfies $G_x\cap K=K_x$.

Since the $G$-action is defined in terms of the $K$-action, the $G$-orbits and the $K$-orbits are clearly the same. On the other hand, for general $x,x'\in i\mathfrak{k}^*$ the orbits $\action{K_x}{x'}$ and $\action{G_x}{x'}$ are not the same. We will need the following condition for equality. 
\begin{lemma}\label{lem:KxorbitGxinvariant}
Suppose that $x,x'\in\action{h}{\positivechamber}$ for some $h\in K$. Then $\action{K_x}{x'}=\action{G_x}{x'}$. In particular, for every $x\in\action{h}{\positivechamber}$ there is a neighbourhood in $\action{h}{i\mathfrak{t}^*}$ where the orbits under $K_x$ and $G_x$ agree.
\end{lemma}
\begin{proof}
It is clear that $\action{K_x}{x'}\subseteq\action{G_x}{x'}$ since $K_x\le G_x$. For the other direction, let $x_0=\action{h^{-1}}{x}$ and $x'_0=\action{h^{-1}}{x'}$ (so $x_0,x_0'\in\positivechamber$) and let $g\in G_x$. We have
\begin{equation}
\action{g}{x'}=\action{k(gh)}{x'_0}=\action{hk(h^{-1}gh)h^{-1}h}{x'_0}=\action{hk(h^{-1}gh)h^{-1}}{x'}.
\end{equation}
We have $h^{-1}gh\in G_{x_0}$, therefore $k(h^{-1}gh)\in K_{x_0}$ (since $B$ fixes $x_0$). This implies $hk(h^{-1}gh)h^{-1}\in K_x$, so $\action{g}{x'}\in\action{K_x}{x'}$.

The second statement follows from the fact that a sufficiently small neighbourhood of $x$ in $\action{h}{i\mathfrak{t}^*}$ only intersects those Weyl chambers whose closure contains $x$, and these can be moved into $\action{h}{\positivechamber}$ with an element of $K_x$.
\end{proof}

We turn to the constant factor in \eqref{eq:GactionHWorbit}. Again, there is at most one extension, which should map $(g,\action{h}{x_0})$ to $e^{2\pairing{x_0}{\alpha(gh)}}$ and our next goal is to verify that this is well defined.
\begin{lemma}\label{lem:alphaproductinvariant}
Let $x_0\in \positivechamber$, $g\in G$ and $u\in K_{x_0}$. Then $\pairing{x_0}{\alpha(g)}=\pairing{x_0}{\alpha(gu)}$.

In particular, if $h_1,h_2\in K$ such that $x=\action{h_1}{x_0}=\action{h_2}{x_0}$ then $\pairing{x_0}{\alpha(gh_1)}=\pairing{x_0}{\alpha(gh_2)}$.
\end{lemma}
\begin{proof}
$K_{x_0}$ is generated by $T$ and the infinitesimal generators $\mathfrak{k}\cap(\mathfrak{g}_\omega+\mathfrak{g}_{-\omega})$ where $\omega\in i\mathfrak{k}^*$ are those simple roots which are orthogonal to $x_0$ (with respect to any $\Ad$-invariant inner product) and $\mathfrak{g}_\omega$ is the corresponding root space.

Let $u\in T$. $T$ normalises $N$ and commutes with $A$, therefore if $g=kan$ is the Iwasawa decomposition of $g$, then $gu=kua(u^{-1}nu)$ is the Iwasawa decomposition of $gu$, so $\alpha(gu)=\alpha(g)$.

Let $\eta\in\mathfrak{g}_{\omega}$. We wish to know how the $\alpha$-component of $g$ changes under right multiplication in the direction of $\eta$. The Iwasawa decomposition gives a diffeomorphism between $K\times A\times N$ and $G$, therefore there exist uniquely $\xi\in\mathfrak{k}$, $\beta\in\mathfrak{a}$ and $\nu\in\mathfrak{n}$ such that
\begin{equation}
\left.\frac{\ed}{\ed s}k\exp(\alpha)n\exp(s\eta)\right|_{s=0}=\left.\frac{\ed}{\ed s}k\exp(s\xi)\exp(\alpha+s\beta)n\exp(s\nu)\right|_{s=0},
\end{equation}
where $g=k\exp(\alpha)n$. Let $L_g:G\to G$ be the left translation by $g$ and $T_eL_g$ its derivative at the identity. We calculate both sides of the above equation as
\begin{equation}
(T_eL_g)\eta=(T_eL_g)\left[\Ad_{n^{-1}\exp(-\alpha)}\xi+\Ad_{n^{-1}}\beta+\nu\right].
\end{equation}
From this we read off $\Ad_{\exp(\alpha)n}\eta=\xi+\Ad_{\exp(\alpha)}\beta+\Ad_{\exp(\alpha)n}\nu=\xi+\beta+\Ad_{\exp(\alpha)n}\nu$ since $A$ acts trivially on $\mathfrak{a}$. This is the Iwasawa decomposition (on the Lie-algebra level), since $\Ad_{\exp(\alpha)n}\nu\in\mathfrak{n}$ (we use that $A$ normalises $N$), therefore $\beta=\left.\frac{\ed}{\ed s}\alpha(g\exp(s\eta))\right|_{s=0}$ is the $\mathfrak{a}$-component of the Iwasawa decomposition of $\Ad_{\exp(\alpha)n}\eta$. We have
\begin{equation}
\Ad_{\exp(\alpha)n}\eta\in\mathfrak{g}_{\omega}\oplus\bigoplus_{\delta}[\mathfrak{g}_{\delta},\mathfrak{g}_{\omega}]\subseteq\mathfrak{g}_{\omega}\oplus\bigoplus_{\delta}\mathfrak{g}_{\delta+\omega},
\end{equation}
where the sum is over the positive roots. If $\omega$ is positive then the $\mathfrak{a}$-component vanishes, whereas if $-\omega$ is a simple positive root then the $\mathfrak{a}$-component is in $[\mathfrak{g}_{-\omega},\mathfrak{g}_{\omega}]$. But this subspace is annihilated by $x_0$ if $\omega$ is orthogonal to $x_0$.

For the second statement, note that the condition $\action{h_1}{x_0}=\action{h_2}{x_0}$ is equivalent to $x_0=\action{h_1^{-1}h_2}{x_0}$, i.e. $h_1^{-1}h_2\in K_{x_0}$. Using the first part with $gh_1$ and $u=h_1^{-1}h_2$ we get $\pairing{x_0}{\alpha(gh_1)}=\pairing{x_0}{\alpha(gh_1h_1^{-1}h_2)}=\pairing{x_0}{\alpha(gh_2)}$ as claimed.
\end{proof}
This justifies the following definition.
\begin{definition}\label{def:chi}
For $x\in i\mathfrak{k}^*$ we define the map $\chi_x:G\to(0,\infty)$ as follows. If $x=\action{h}{x_0}$ (with $h\in K$ and $x_0\in \positivechamber$) and $g\in G$ then we set
\begin{equation}
\chi_x(g)=e^{2\pairing{x_0}{\alpha(gh)}}.
\end{equation}
\end{definition}
In terms of this map we may rewrite \eqref{eq:GactionHWorbit} as
\begin{equation}\label{eq:GactionHWorbitchi}
\pi_\lambda(g)\ketbra{\action{h}{v_\lambda}}{\action{h}{v_\lambda}}\pi_\lambda(g)^* = \chi_{\lambda}(g)\ketbra{\action{k(gh)}{v_\lambda}}{\action{k(gh)}{v_\lambda}}.
\end{equation}

\begin{proposition}
The map $\chi$ is multiplicative in the following sense. If $x\in i\mathfrak{k}^*$ and $g_1,g_2\in G$, then
\begin{equation}
\chi_x(g_2g_1)=\chi_{\action{g_1}{x}}(g_2)\chi_x(g_1).
\end{equation}
\end{proposition}
\begin{proof}
Write $x=\action{h}{x_0}$ with $x_0\in i\mathfrak{k}^*$ and $h\in K$. In the same way as in the proof of~\ref{it:Gaction} in Proposition~\ref{prop:welldefinedaction}, if $g_1h=k_1a_1n_1$ and $g_2k(g_1h)=g_2k_1=k_2a_2n_2$ then the $A$-component of $g_2g_1h$ is $a_2a_1$, therefore
\begin{equation}
\alpha(g_2g_1h)=\alpha(g_2k(g_1h))+\alpha(g_1h).
\end{equation}
Using this we compute
\begin{equation}
\begin{split}
\chi_x(g_2g_1)
 & = e^{2\pairing{x_0}{\alpha(g_2g_1h)}}  \\
 & = e^{2\pairing{x_0}{\alpha(g_2k(g_1h))+\alpha(g_1h)}}  \\
 & = e^{2\pairing{x_0}{\alpha(g_2k(g_1h))}}e^{2\pairing{x_0}{\alpha(g_1h)}}  \\
 & = \chi_{\action{g_1}{x}}(g_2)\chi_x(g_1).
\end{split}
\end{equation}
\end{proof}
In particular, the restriction of $\chi_x$ to $G_x$ is a one dimensional character. From $\chi_x(e)=1$ and multiplicativity it follows that $\chi_x(g^{-1})=\chi_{\action{g^{-1}}{x}}(g)^{-1}$.

\begin{proposition}\label{prop:continuity}
Let $K$ act on the space $i\mathfrak{k}^*\times G$ as $\action{k}{(x,g)}=(\action{k}{x},gk^{-1})$. Then $\chi:i\mathfrak{k}^*\times G\to\mathbb{R}$ is continuous and invariant.
\end{proposition}
\begin{proof}
Consider the closed set $\positivechamber\times G\subseteq i\mathfrak{k}^*\times G$. On this set the function simplifies to $(x,g)\mapsto 2\pairing{x}{\alpha(g)}=:\varphi(x,g)$, which is continuous. If $(x_0,g)\in \positivechamber\times G$ and $k\in K$ such that $\action{k}{(x_0,g)}\in \positivechamber\times G$, then $x_0=\action{k}{x_0}$, i.e. $k\in K_{x_0}$, therefore by Lemma~\ref{lem:alphaproductinvariant} we have
\begin{equation}
\varphi(\action{k}{(x_0,g)})=\varphi(x_0,gk^{-1})=2\pairing{x_0}{\alpha(gk^{-1})}=2\pairing{x_0}{\alpha(g)}=\varphi(x_0,g).
\end{equation}
Since $\action{K}{(\positivechamber\times G)}=i\mathfrak{k}^*\times G$, there is a unique $K$-invariant extension, continuous by \cite[3.3. Theorem in Chapter I]{bredon1972introduction}. To see that this extension is $\chi$ we need to verify that $\chi$ is $K$-invariant. Let $x_0\in\positivechamber$, $k,h\in K$, $x=\action{h}{x_0}$ and $g\in G$. Then
\begin{equation}
\chi_{\action{k}{x}}(gk^{-1})=\chi_{\action{kh}{x_0}}(gk^{-1})=e^{2\pairing{x_0}{\alpha(gk^{-1}kh)}}=e^{2\pairing{x_0}{\alpha(gh)}}=\chi_{x}(g).
\end{equation}
\end{proof}

We conclude this section with a closely related real valued map on $i\mathfrak{k}^*\times i\mathfrak{k}$ that can be viewed as a non-bilinear modification of the duality pairing, and reduces to it when $K$ is abelian. This map will appear in one of the equivalent expressions for the rate function below.
\begin{definition}
We define a map $\nonlinearpairing{\cdot}{\cdot}:i\mathfrak{k}^*\times i\mathfrak{k}\to\mathbb{R}$ as follows. Let $x\in i\mathfrak{k}^*$ and $\xi\in i\mathfrak{k}$. Write $x=\action{h}{x_0}$ where $h\in K$ and $x_0\in \positivechamber$, and set
\begin{equation}
\nonlinearpairing{x}{\xi}=-\ln\chi_x(\exp(-\xi/2))=-2\pairing{x_0}{\alpha\left(\exp\left(-\action{h^{-1}}{\xi}/2\right)\right)}.
\end{equation}
\end{definition}

\begin{example}
Let $K=SU(2)$. Then $i\mathfrak{su}(2)$ and its dual can be identified with the space of traceless hermitian $2\times 2$ matrices. Let $x$ and $\xi$ be such matrices. Then
\begin{equation}
\nonlinearpairing{x}{\xi}=-2\norm[\infty]{x}\ln\left[\cosh\norm[\infty]{\xi}-\frac{\Tr x\xi}{2\norm[\infty]{x}\norm[\infty]{\xi}}\sinh\norm[\infty]{\xi}\right].
\end{equation}
\end{example}

\subsection{Rate function and some properties}\label{sec:ratefunction}

We now give the definition of the rate function appearing in Theorem~\ref{thm:LDP}, then we prove its equivalence with several other expressions.
\begin{definition}\label{def:ratefunctiondef}
For $\rho\ge 0$ (not necessarily normalised), we define the function $\ratefunction{\rho}:i\mathfrak{k}^*\to(-\infty,\infty]$ as
\begin{equation}\label{eq:ratefunctiondef}
\ratefunction{\rho}(x)=\sup_{g\in G}-\ln\chi_x(g^{-1})-\ln\Tr\pi(g)^*\rho\pi(g).
\end{equation}
\end{definition}
\begin{proposition}\label{prop:ratefunctionexpressions}
Let $x=\action{h}{x_0}$ with $x_0\in \positivechamber$ and $h\in K$. Then
\begin{subequations}
\begin{align}
\ratefunction{\rho}(x)
 & = \sup_{g\in G_x}\ln\chi_x(g)-\ln\Tr\pi(g)^*\rho\pi(g)  \label{eq:alternativeexpressionGx}  \\
 & = \sup_{b\in hBh^{-1}}\ln\chi_x(b)-\ln\Tr\pi(b)^*\rho\pi(b)  \label{eq:alternativeexpressionB}  \\
 & = \sup_{\alpha\in\mathfrak{a}}\sup_{n\in N}\pairing{x_0}{\alpha}-\ln\Tr\pi(n)^*\pi(\exp\alpha/2)\pi(h)^*\rho\pi(h)\pi(\exp\alpha/2)\pi(n)  \label{eq:alternativeexpressionAN}  \\
 & = \sup_{\alpha\in\mathfrak{a}}\sup_{n\in N}\pairing{x_0}{\alpha}-\ln\Tr\pi(\exp\alpha/2)\pi(n)^*\pi(h)^*\rho\pi(h)\pi(n)\pi(\exp\alpha/2)  \label{eq:alternativeexpressionNA}  \\
 & = \sup_{\xi\in i\mathfrak{k}}\nonlinearpairing{x}{\xi}-\ln Z_\rho(\xi)  \label{eq:alternativeexpressionZ}
\end{align}
\end{subequations}
where
\begin{equation}\label{eq:momentgenerating}
Z_\rho(\xi)=\Tr\rho\pi(\exp\xi).
\end{equation}
\end{proposition}
\begin{proof}
$\chi_x$ is multiplicative on $G_x$, therefore $-\ln\chi_x(g^{-1})=\ln\chi_x(g)$ for $g\in G_x$. From this equality and $G\ge G_x\ge hBh^{-1}=hANh^{-1}hTh^{-1}=hNAh^{-1}hTh^{-1}$ we obtain the inequalitites
\begin{equation}
\begin{split}
\ratefunction{\rho}(x)
 & \ge \sup_{g\in G_x}\ln\chi_x(g)-\ln\Tr\pi(g)^*\rho\pi(g)  \\
 & \ge \sup_{b\in hBh^{-1}}\ln\chi_x(b)-\ln\Tr\pi(b)^*\rho\pi(b)  \\
 & = \sup_{\alpha\in\mathfrak{a}}\sup_{n\in N}\pairing{x_0}{\alpha}-\ln\Tr\pi(n)^*\pi(\exp\alpha/2)\pi(h)^*\rho\pi(h)\pi(\exp\alpha/2)\pi(n)  \\
 & = \sup_{\alpha\in\mathfrak{a}}\sup_{n\in N}\pairing{x_0}{\alpha}-\ln\Tr\pi(\exp\alpha/2)\pi(n)^*\pi(h)^*\rho\pi(h)\pi(n)\pi(\exp\alpha/2).
\end{split}
\end{equation}

For $\xi\in i\mathfrak{k}$ let $\exp(-\action{h^{-1}}{\xi/2})=k_0a_0n_0$ be the Iwasawa decomposition. Let $\alpha\in\mathfrak{a}$ such that $\exp(-\alpha/2)=a_0$ and $n=n_0^{-1}$. Then
\begin{equation}
\begin{split}
& \pairing{x_0}{\alpha}-\ln\Tr\pi(\exp\alpha/2)\pi(n)^*\pi(h)^*\rho\pi(h)\pi(n)\pi(\exp\alpha/2)  \\
 & = -2\pairing{x_0}{\alpha(\exp(-\action{h^{-1}}{\xi/2}))}-\ln\Tr\pi(a_0^{-1})\pi(n_0^{-1})^*\pi(h)^*\rho\pi(h)\pi(n_0^{-1})\pi(a_0^{-1})  \\
 & = \nonlinearpairing{x}{\xi}-\ln\Tr\rho\pi(h)\pi(n_0^{-1})\pi(a_0^{-1})\pi(k_0^{-1})\pi(k_0^{-1})^*\pi(a_0^{-1})\pi(n_0^{-1})^*\pi(h)^*  \\
 & = \nonlinearpairing{x}{\xi}-\ln\Tr\rho\pi(h)\pi(k_0a_0n_0)^{-1}{\pi(k_0a_0n_0)^{-1}}^*\pi(h)^*  \\
 & = \nonlinearpairing{x}{\xi}-\ln\Tr\rho\pi(h)\exp(\action{h^{-1}}{\xi})\pi(h)^*  \\
 & = \nonlinearpairing{x}{\xi}-\ln\Tr\rho\pi(\exp(\xi))  \\
 & = \nonlinearpairing{x}{\xi}-\ln Z_\rho(\xi),
\end{split}
\end{equation}
where we used that $(k_0a_0n_0)^*=k_0a_0n_0$. Therefore
\begin{multline}
\sup_{\alpha\in\mathfrak{a}}\sup_{n\in N}\pairing{x_0}{\alpha}-\ln\Tr\pi(\exp\alpha/2)\pi(n)^*\pi(h)^*\rho\pi(h)\pi(n)\pi(\exp\alpha/2)  \\  \ge\sup_{\xi\in i\mathfrak{k}}\nonlinearpairing{x}{\xi}-\ln Z_\rho(\xi).
\end{multline}

For $g\in G$ we have $gg^*\in P$, therefore there is a unique $\xi\in i\mathfrak{k}$ such that $gg^*=\exp\xi$. Let $g=pu$ be the polar decomposition with $p=\exp(\xi/2)$ and $u\in K$. Then
\begin{equation}
\begin{split}
\nonlinearpairing{x}{\xi}-\ln Z_\rho(\xi)
 & = -\ln\chi_x(\exp(-\xi/2))-\ln\Tr\rho\pi(\exp(\xi))  \\
 & = -\ln\chi_x(p^{-1})-\ln\Tr\pi(g)^*\rho\pi(g)  \\
 & = -\ln\chi_x(ug^{-1})-\ln\Tr\pi(g)^*\rho\pi(g)  \\
 & = -\ln\chi_x(g^{-1})-\ln\Tr\pi(g)^*\rho\pi(g),
\end{split}
\end{equation}
in the last step using that $\chi_x$ is invariant under left multiplication with elements of $K$ (as $u$ only contributes to the $K$-component). Therefore
\begin{equation}
\sup_{\xi\in i\mathfrak{k}}\nonlinearpairing{x}{\xi}-\ln Z_\rho(\xi)\ge\sup_{g\in G}-\ln\chi_x(g^{-1})-\ln\Tr\pi(g)^*\rho\pi(g)=\ratefunction{\rho}(x).
\end{equation}
\end{proof}
\begin{remark}
When $K=U(1)^d$, then \eqref{eq:momentgenerating} reduces to the moment generating function and $\nonlinearpairing{\cdot}{\cdot}$ becomes the usual pairing between a vector space and its dual. Therefore \eqref{eq:alternativeexpressionZ} can be viewed as a nonabelian generalisation of the Legendre--Fenchel transform of the logarithmic moment generating function.
\end{remark}

The following lemma shows that the supremum over $N$ in \eqref{eq:alternativeexpressionAN} and \eqref{eq:alternativeexpressionNA} can be replaced with a maximum as long as we first optimise over $N$ and then over $\mathfrak{a}$.
\begin{lemma}\label{lem:unipotentminlength}
Let $N$ be a unipotent algebraic group, $\pi:N\to\GL(\mathcal{H})$ a representation on a finite dimensional Hilbert space and $\rho\in\states(\mathcal{H})$. Then the function $n\mapsto\Tr\pi(n)^*\rho\pi(n)$ has a minimum.
\end{lemma}
\begin{proof}
$\tilde{N}:=\pi(N)^*$ is also a unipotent algebraic group. Let $\psi\in\mathcal{H}\otimes\mathbb{C}^d$ be a purification of $\rho$ and consider the representation $\tilde{n}\mapsto\tilde{n}\otimes I$ of $\tilde{N}$. By \cite[Proposition of 2.5]{steinberg2006conjugacy} the orbit of $\psi$ under this action is closed. Therefore there is a vector in the orbit of minimal length. The claim follows since $\Tr\tilde{n}\rho\tilde{n}^*=\norm{(\tilde{n}\otimes I)\psi}^2$.
\end{proof}

\begin{example}[{Duffield, \cite[Theorem 2.1.]{duffield1990large}}]\label{ex:maximallymixed}
Let $K$ be arbitrary compact connected, $\pi:K\to U(\mathcal{H})$ irreducible and $\rho=\frac{I}{\dim\mathcal{H}}$. Then
\begin{equation}
\begin{split}
\ratefunction{\rho}(\action{h}{x_0})
 & = \sup_{\alpha\in\mathfrak{a}}\max_{n\in N}\pairing{x_0}{\alpha}-\ln\Tr\pi(n)^*\pi(\exp\alpha/2)\pi(h)^*\frac{I}{\dim\mathcal{H}}\pi(h)\pi(\exp\alpha/2)\pi(n)  \\
 & = \sup_{\alpha\in\mathfrak{a}}\max_{n\in N}\pairing{x_0}{\alpha}-\ln\frac{\Tr\pi(n)^*\pi(\exp\alpha)\pi(n)}{\dim\mathcal{H}}.
\end{split}
\end{equation}
The maximum over $n$ is attained at $n=e$, as can be seen by computing the trace in a basis where $T_e\pi(\alpha)$ is diagonal and $\pi(N)$ consists of upper triangular matrices with $1$ on their diagonals. The formula is $K$-invariant, therefore the infimum over $K$ in \eqref{eq:contraction} can be omitted, leading to the formula
\begin{equation}\label{eq:multiplicityLDP}
\tilde{I}_\rho(x_0)=\ratefunction{\rho}(\action{h}{x_0})=\sup_{\alpha\in\mathfrak{a}}\pairing{x_0}{\alpha}-\ln\frac{\Tr\pi(\exp\alpha)}{\dim\mathcal{H}}.
\end{equation}
The character $\Tr\pi(\cdot)$ is $K$-invariant and $\pairing{x_0}{\alpha}$ is maximal within the $K$-orbit of $\alpha$ when $\alpha$ is in the image of the dominant Weyl chamber under the identification of $i\mathfrak{t}$ with $i\mathfrak{t}^*$ via an invariant inner product. Therefore the supremum can be restricted to this subset of $\mathfrak{a}$.
\end{example}
For a more precise asymptotic formula for the multiplicities see \cite[Theorem 9.]{tate2004lattice}.

\begin{example}[Ceg\l a--Lewis--Raggio, \cite{cegla1988free}]
Let $K=SU(2)$, $\pi:K\to U(\mathbb{C}^{2j+1})$ a spin-$j$ irreducible representation and $\rho=\frac{I}{2j+1}$. In this case $\positivechamber\simeq[0,\infty)$ and $\mathfrak{a}\simeq\mathbb{R}$. Specializing \eqref{eq:multiplicityLDP} to this case we get
\begin{equation}
\begin{split}
\tilde{I}_\rho(x_0)
 & = \ln(2j+1)+\sup_{\alpha\in\mathfrak{a}}\pairing{x_0}{\alpha}-\ln\Tr\pi(\exp\alpha)  \\
 & = \ln(2j+1)+\sup_{\alpha\ge0}x_0\alpha-\ln\sum_{l=-j}^j\exp(l\alpha)  \\
 & = \ln(2j+1)+\sup_{\alpha\ge0}x_0\alpha-\ln\frac{\sinh\frac{(2j+1)\alpha}{2}}{\sinh\frac{\alpha}{2}}.
\end{split}
\end{equation}
\end{example}

\begin{example}[Cram\'er, \cite{cramer1938nouveau}]
Let $K=U(1)^d$ and $\pi:K\to U(\mathcal{H})$ a finite dimensional unitary representation. Let $\rho$ be arbitrary and write $r_n=\Tr P_n\rho$ where $P_n$ is the orthogonal projection corresponding to the irreducible representation labelled by $n\in\mathbb{Z}^d$ (see Example~\ref{ex:momentmap:torus}). If $X_1,X_2,\ldots,X_m$ are independent and identically distributed discrete vector random variables that take the value $n$ with probability $r_n$, then $\mu_m$ is the distribution of $\frac{1}{m}(X_1+X_2+\cdots+X_m)$.

To compute the rate function we use $\pi(\exp\alpha)=\sum_{n\in\mathbb{Z}^d}e^{\pairing{n}{\alpha}}P_n$:
\begin{equation}
\ratefunction{\rho}(x)=\sup_{\alpha\in\mathbb{R}^d}\pairing{x}{\alpha}-\ln\sum_{n\in\mathbb{Z}^d}r_ne^{\pairing{n}{\alpha}},
\end{equation}
which is the Legendre--Fenchel transform of the logarithm of the moment generating function of $X_i$.
\end{example}

\begin{example}[Keyl, {\cite[Theorem 3.2]{keyl2006quantum}}]\label{ex:stateestimation}
Let $K=U(d)$, $\pi$ the identity map (the standard representation on $\mathbb{C}^d$) and $\rho$ an arbitrary state.
Let us first minimise $\Tr\pi(n)^*\sigma\pi(n)$ over $n\in N$. With $N$ the set of upper triangular unipotent matrices (see Example~\ref{ex:Borel:unitary}), it is possible to choose $n$ in such a way that $\pi(n)^*\sigma\pi(n)$ is diagonal. As in Example~\ref{ex:maximallymixed}, this is where the minimum is attained. To find the diagonal entries, note that the principal minors are invariant under this action of $N$. If $\principalminor_j(\sigma)$ denotes the determinant of the upper left $j\times j$ submatrix of $\sigma$ (with the convention $\principalminor_0(\sigma)=1$), then the resulting diagonal matrix is
\begin{equation}
\begin{bmatrix}
\frac{\principalminor_1(\sigma)}{\principalminor_0(\sigma)} & 0 & \cdots & 0  \\
0 & \frac{\principalminor_2(\sigma)}{\principalminor_1(\sigma)} & \ddots & \vdots  \\
\vdots & \ddots & \ddots & 0  \\
0 & \cdots & 0 & \frac{\principalminor_d(\sigma)}{\principalminor_{d-1}(\sigma)}.
\end{bmatrix}
\end{equation}
Let $\alpha_1,\ldots,\alpha_d\in\mathbb{R}$ and $x_{0,1},\ldots,x_{0,d}\in\mathbb{R}$ be the diagonal entries of $\alpha\in\mathfrak{a}\simeq\mathbb{R}^d$ and $x_0\in \positivechamber$. With $\sigma=\pi(\exp(\alpha/2))\pi(h)^*\rho\pi(h)\pi(\exp(\alpha/2))$ the rate function formula \eqref{eq:alternativeexpressionAN} becomes
\begin{equation}
\ratefunction{\rho}(\action{h}{x_0})=\sup_{\alpha\in\mathbb{R}^d}\pairing{x_0}{\alpha}-\ln\sum_{i=1}^d e^{\alpha_i}\frac{\principalminor_i(\pi(h)^*\rho\pi(h))}{\principalminor_{i-1}(\pi(h)^*\rho\pi(h))}
\end{equation}
The supremum can be found by differentiation, which gives
\begin{equation}
\ratefunction{\rho}(\action{h}{x_0})=\begin{cases}
\displaystyle\sum_{i=1}^d\left[x_{0,i}\ln x_{0,i}-x_{0,i}\left(\ln\frac{\principalminor_i(\pi(h)^*\rho\pi(h))}{\principalminor_{i-1}(\pi(h)^*\rho\pi(h))}\right)\right] & \parbox{15em}{if $x_{0,1}+\cdots+x_{0,d}=1$  \\  and $\forall i:x_{0,i}\ge 0$}  \\
\infty & \text{otherwise.}
\end{cases}
\end{equation}

Using the contraction principle one recovers the result of Keyl and Werner \cite{keyl2001estimating} by the same reasoning as in \cite[Proof of Lemma 4.15]{keyl2006quantum}.
\end{example}

\begin{example}
Let $K=U(d_1)\times U(d_2)$, $\pi:K\to U(\mathbb{C}^{d_1}\otimes\mathbb{C}^{d_2})$ the tensor product of the standard representations and let $\rho=\ketbra{\psi}{\psi}$ be a pure state. We identify $\psi$ with a $d_1\times d_2$ matrix. Under this identification, the action of $K$ becomes left multiplication with the first unitary and right multiplication with the transpose of the second one. Similarly as before, we choose $N$ to be group of pairs of upper triangular unipotent matrices. Let $x_{1,0},\alpha_1\in\mathbb{R}^{d_1}$, $x_{2,0},\alpha_2\in\mathbb{R}^{d_2}$, identified with diagonal matrices of sizes $d_1\times d_1$ and $d_2\times d_2$ and with $x_{1,0}$ and $x_{2,0}$ nonincreasing. For $(h_1,h_2)\in K$, the pair $(\action{h_1}{x_{1,0}},\action{h_2}{x_{2,0}})$ can be viewed as an element of $i\mathfrak{k}$. As in Example~\ref{ex:stateestimation}, the supremum over $(n_1,n_2)\in N$ in \eqref{eq:alternativeexpressionAN} is attained when $n_1^*\exp(\alpha_1/2)h_1^*\psi\overline{h_2}\exp(\alpha_2/2)\overline{n_2}$ is diagonal, and the diagonal form can be determined using the invariance of the principal minors under the $N$-action. The rate function is therefore
\begin{multline}
\ratefunction{\ketbra{\psi}{\psi}}(\action{h_1}{x_{1,0}},\action{h_2}{x_{2,0}})=\sup_{\substack{\alpha_1\in\mathbb{R}^{d_1}  \\  \alpha_2\in\mathbb{R}^{d_2}}}\pairing{x_{1,0}}{\alpha_1}+\pairing{x_{2,0}}{\alpha_2}  \\  -\ln\sum_{i=1}^{\min\{d_1,d_2\}} e^{\alpha_{1,i}+\alpha_{2,i}}\left|\frac{\principalminor_i(h_1^*\psi\overline{h_2})}{\principalminor_{i-1}(h_1^*\psi\overline{h_2})}\right|^2.
\end{multline}
The supremum is infinite if $x_{1,0}$ differs from $x_{2,0}$ up to trailing zeros, or any of the entires is negative, or if the vectors do not sum to one. Otherwise it evaluates to
\begin{equation}
\ratefunction{\ketbra{\psi}{\psi}}(\action{h_1}{x_0},\action{h_2}{x_0})=\sum_{i=1}^{\min\{d_1,d_2\}}\left[x_{0,i}\ln x_{0,i}-x_{0,i}\ln\left|\frac{\principalminor_i(h_1^*\psi\overline{h_2})}{\principalminor_{i-1}(h_1^*\psi\overline{h_2})}\right|^2\right]
\end{equation}
\end{example}

Next we prove some properties of $\ratefunction{\rho}$.
\begin{proposition}\label{prop:lowersemicontinuous}
$\ratefunction{\rho}(x)$ is lower semicontinuous.
\end{proposition}
\begin{proof}
$\ratefunction{\rho}(x)$ is the supremum of the family of continuous functions $x\mapsto-\ln\chi_x(g^{-1})-\ln\Tr\pi(g)^*\rho\pi(g)$, hence lower semicontinuous.
\end{proof}

\begin{proposition}\label{prop:ratefunctionnonnegative}
Let $\rho\in\states(\mathcal{H})$. $\ratefunction{\rho}(x)\ge 0$ for every $x\in i\mathfrak{k}^*$ and if $x\neq J(\rho)$ then $\ratefunction{\rho}(x)>0$.
\end{proposition}
\begin{proof}
Since $e\in G_x$,
\begin{equation}
\begin{split}
\ratefunction{\rho}(x)
 & = \sup_{g\in G_x}\ln\chi_x(g)-\ln\Tr\pi(g)^*\rho\pi(g)  \\
 & \ge \ln\chi_x(e)-\ln\Tr\pi(e)^*\rho\pi(e)  \\
 & = -\ln\Tr\rho=0.
\end{split}
\end{equation}

$G_x$ is a smooth manifold, the expression $\ln\chi_x(g)-\ln\Tr\pi(g)^*\rho\pi(g)$ is a smooth function of $g\in G_x$ and is zero for $g=e$. It follows that if $e$ is not a critical point then the supremum is strictly positive. The tangent space decomposes as $T_eG_x=T_eK_x\oplus(\action{h}{\mathfrak{a}})\oplus(\action{h}{\mathfrak{n}})$ where $h\in K$ is any element such that $x=\action{h}{x_0}$, $x_0\in \positivechamber$.

Let $\beta\in\action{h}{\mathfrak{a}}$. Then
\begin{multline}
\left.\frac{\ed}{\ed s}\ln\chi_x(\exp s\beta)-\ln\Tr\pi(\exp s\beta)^*\rho\pi(\exp s\beta)\right|_{s=0}  \\
 = 2\pairing{x_0}{\action{h^{-1}}{\beta}}-\frac{\Tr T_e\pi(\beta)\rho+\rho T_e\pi(\beta)}{\Tr\rho}
 = 2\pairing{x}{\beta}-2\pairing{J(\rho)}{\beta}.
\end{multline}

Let $\nu\in \action{h}{\mathfrak{n}}$. Then
\begin{multline}
\left.\frac{\ed}{\ed s}\ln\chi_x(\exp s\nu)-\ln\Tr\pi(\exp s\nu)^*\rho\pi(\exp s\nu)\right|_{s=0}  \\
 = -\frac{\Tr T_e\pi(\nu)\rho+\rho T_e\pi(\nu)}{\Tr\rho}
 = -2\pairing{J(\rho)}{\nu}.
\end{multline}

If $e$ is a critical point then both derivatives vanish for every $\beta$ and $\nu$, thus $J(\rho)$ and $x$ agree on $\action{h}{\mathfrak{b}}$. Since they both vanish on $\mathfrak{k}$, this means $x=J(\rho)$.
\end{proof}

Recall that the domain of an extended real valued function $f:X\to(-\infty,\infty]$ is the set $\domain f=\setbuild{x\in X}{f(x)<\infty}$. We now show that the domain of the rate function is precompact.
\begin{proposition}\label{prop:precompactdomain}
Let $\Delta\subseteq i\mathfrak{t}^*$ be the convex hull of the set of weights appearing in the decomposition of $\mathcal{H}$ with respect to the action of $T$. Then for every $x\notin \action{K}{\Delta}$ we have $\ratefunction{\rho}(x)=\infty$.
\end{proposition}
\begin{proof}
Suppose that $x\notin\action{K}{\Delta}$ and write $x=\action{h}{x_0}$ with $x_0\in \positivechamber$, $h\in K$. Since $\Delta$ is compact and convex, there is a hyperplane in $i\mathfrak{t}^*$ separating $\Delta$ and $x_0$, so there is an element $\beta\in\mathfrak{t}$ such that $\pairing{x_0}{\beta}>\max_{x'\in\Delta}\pairing{x'}{\beta}$.

We use \eqref{eq:alternativeexpressionGx} and that $G_x$ contains $hAh^{-1}$, therefore for any $s\in\mathbb{R}$ we have
\begin{equation}
\begin{split}
\ratefunction{\rho}(x)
 & \ge \ln\chi_x(h\exp(s\beta)h^{-1})-\ln\Tr\pi(h\exp(s\beta)h^{-1})^*\rho\pi(h\exp(s\beta)h^{-1})  \\
 & = 2\pairing{x_0}{s\beta}-\ln\Tr\pi(h)^*\rho\pi(h)\pi(\exp(2s\beta))  \\
 & \ge 2\pairing{x_0}{s\beta}-\max_{x'\in\Delta}\pairing{x'}{2s\beta}  \\
 & = 2s\left(\pairing{x_0}{\beta}-\max_{x'\in\Delta}\pairing{x'}{\beta}\right).
\end{split}
\end{equation}
The coefficient of $s$ is strictly positive, therefore letting $s\to\infty$ shows that $\ratefunction{\rho}(x)=\infty$.
\end{proof}

\begin{proposition}\label{prop:ratesupport}
Let $\rho,\sigma\ge 0$ and $p\ge 0$ such that $\rho\le p\sigma$. For every $x\in i\mathfrak{k}^*$ we have
\begin{equation}
\ratefunction{\rho}(x)\ge\ratefunction{\sigma}(x)-\ln p.
\end{equation}
In particular, if $\support\rho\le\support\sigma$ then $\domain\ratefunction{\rho}\subseteq\domain\ratefunction{\sigma}$.
\end{proposition}
\begin{proof}
\begin{equation}
\begin{split}
\ratefunction{\rho}(x)
 & = \sup_{g\in G}-\ln\chi_x(g^{-1})-\ln\Tr\pi(g)^*\rho\pi(g)  \\
 & \ge \sup_{g\in G}-\ln\chi_x(g^{-1})-\ln\Tr\pi(g)^*p\sigma\pi(g)  \\
 & \ge \sup_{g\in G}-\ln\chi_x(g^{-1})-\ln\Tr\pi(g)^*\sigma\pi(g)-\ln p
   = \ratefunction{\sigma}(x)-\ln p.
\end{split}
\end{equation}
The second statement follows since $\support\rho\le\support\sigma$ iff $\rho\le p\sigma$ for some $p>0$ (this uses $\dim\mathcal{H}<\infty$).
\end{proof}

\begin{proposition}\label{prop:invertiblecontinuity}
Suppose that $\rho\in\states(\mathcal{H})$ is an invertible state. Then $\domain\ratefunction{\rho}=\domain\ratefunction{I}$ (the rate function for the identity operator $I$) is $K$-invariant and the restriction of $\ratefunction{\rho}$ to $\action{K}{\relativeinterior(\positivechamber\cap\domain\ratefunction{\rho})}$ is continuous.
\end{proposition}
\begin{proof}
The statement on the domain is an immediate consequence of Proposition~\ref{prop:ratesupport} and $\support\rho=\mathcal{H}=\support I$, invariance of $I$ and thus of $\ratefunction{I}$.

We prove continuity. Let $f:K\times\relativeinterior(\positivechamber\cap\domain\ratefunction{\rho})\to\mathbb{R}$ be defined as $f(h,x_0)=\ratefunction{\rho}(\action{h}{x_0})=\ratefunction{\pi(h)^*\rho\pi(h)}(x_0)$. For every fixed $h$ the function $f(h,\cdot)$ is convex, as it is the supremum of a family of affine functions (see e.g. \eqref{eq:alternativeexpressionAN}).

Let $\eta\in\mathfrak{k}$ and $\rho_s=\pi(\exp s\eta)^*\rho\pi(\exp s\eta)$, and for $\xi\in i\mathfrak{k}$ let $F_\xi(s)=\ln Z_{\rho_s}(\xi)=\ln\Tr\rho_s\pi(\exp\xi)$. Then
\begin{equation}
F_\xi'(s)=\frac{\Tr[\rho_s,T_e\pi(\eta)]\pi(\exp\xi)}{\Tr\rho_s\pi(\exp\xi)}.
\end{equation}
From
\begin{equation}
\begin{split}
\rho_s^{-1/2}[\rho_s,T_e\pi(\eta)]\rho_s^{-1/2}
 & \le \norm[\infty]{\rho_s^{-1/2}[\rho_s,T_e\pi(\eta)]\rho_s^{-1/2}}I  \\
 & = \norm[\infty]{\rho_s^{1/2}T_e\pi(\eta)\rho_s^{-1/2}-\rho_s^{-1/2}T_e\pi(\eta)\rho_s^{1/2}}  \\
 & \le 2\norm[\infty]{\rho_s^{1/2}}\norm[\infty]{\rho_s^{-1/2}}\norm[\infty]{T_e\pi(\eta)}  \\
 & = 2\norm[\infty]{\rho^{1/2}}\norm[\infty]{\rho^{-1/2}}\norm[\infty]{T_e\pi(\eta)}
\end{split}
\end{equation}
we infer $[\rho_s,T_e\pi(\eta)]\le 2\norm[\infty]{\rho^{1/2}}\norm[\infty]{\rho^{-1/2}}\norm[\infty]{T_e\pi(\eta)}\rho_s$. Since $\pi(\exp\xi)\ge 0$, we have
\begin{equation}
\begin{split}
\ln Z_{\pi(\eta)^*\rho\pi(\eta)}(\xi)-\ln Z_{\rho}(\xi)
 & = F_\xi(1)-F_\xi(0)  \\
 & = \int_0^1F_\xi'(s)\ed s  \\
 & = \int_0^1\frac{\Tr[\rho_s,T_e\pi(\eta)]\pi(\exp\xi)}{\Tr\rho_s\pi(\exp\xi)}\ed s  \\
 & \le 2\norm[\infty]{\rho^{1/2}}\norm[\infty]{\rho^{-1/2}}\norm[\infty]{T_e\pi(\eta)}.
\end{split}
\end{equation}
Using this we get
\begin{equation}
\begin{split}
f(\exp(\eta),x_0)
 & = \ratefunction{\pi(\exp(\eta))^*\rho\pi(\exp(\eta))}(x_0)  \\
 & = \sup_{\xi\in i\mathfrak{k}}\nonlinearpairing{x_0}{\xi}-\ln Z_{\pi(\eta)^*\rho\pi(\eta)}(\xi)  \\
 & \le \sup_{\xi\in i\mathfrak{k}}\nonlinearpairing{x_0}{\xi}-\ln Z_{\rho}(\xi)+2\norm[\infty]{\rho^{1/2}}\norm[\infty]{\rho^{-1/2}}\norm[\infty]{T_e\pi(\eta)}  \\
 & = \ratefunction{\rho}(x_0)+2\norm[\infty]{\rho^{1/2}}\norm[\infty]{\rho^{-1/2}}\norm[\infty]{T_e\pi(\eta)}  \\
 & = f(e,x_0)+2\norm[\infty]{\rho^{1/2}}\norm[\infty]{\rho^{-1/2}}\norm[\infty]{T_e\pi(\eta)}.
\end{split}
\end{equation}
The same reasoning applies to $\pi(h)^*\rho\pi(h)$ instead of $\rho$ as well, with the same constant, therefore $f(\cdot,x_0)$ is continuous on $K$ for every $x_0$.

By \cite[Theorem 10.7]{rockafellar1970convex}, $f$ is jointly continuous. Let $H$ denote the common stabiliser subgroup of the points of $\relativeinterior(\positivechamber\cap\domain\ratefunction{\rho})$, and consider the right action of $H$ on $K$ by translations. $f$ is invariant with respect to this action, therefore descends to a continuous function $\tilde{f}$ on $K\times_H\relativeinterior(\positivechamber\cap\domain\ratefunction{\rho})$, which is homeomorphically mapped to $\action{K}{\relativeinterior(\positivechamber\cap\domain\ratefunction{\rho})}$. Under this homeomorphism $\tilde{f}$ corresponds to the restriction of $\ratefunction{\rho}$ to $\action{K}{\relativeinterior(\positivechamber\cap\domain\ratefunction{\rho})}$, which is therefore continuous.
\end{proof}

\subsection{Upper bound}\label{sec:upperbound}

We are in a position to prove the upper bound part of the large deviation principle. We prove a stronger statement involving any measurable set instead of only closed ones.
\begin{proposition}[Large deviation upper bound]\label{prop:LDPupper}
Let $F\subseteq\mathfrak{i}k^*$ be a measurable set and $m\in\mathbb{N}$. Then
\begin{equation}\label{eq:LDPupper}
\mu_m(F)\le(m+1)^{\dim\mathcal{H}(\dim\mathcal{H}+1)/2}e^{-m\inf_{x\in F}\ratefunction{\rho}(x)},
\end{equation}
therefore
\begin{equation}
\limsup_{m\to\infty}\frac{1}{m}\ln\mu_m(F)\le-\inf_{x\in F}\ratefunction{\rho}(x).
\end{equation}
\end{proposition}
\begin{proof}
Recall that \eqref{eq:muintegral} expresses $\mu_m(F)$ as a sum of integrals:
\begin{equation}
\mu_m(F)
 = \sum_{\lambda}\dim(\mathcal{H}_\lambda)\int_{J_\lambda^{-1}(mF)}\Tr\rho^{\otimes m}p_{\lambda,\mathcal{H}^{\otimes m}}([v_\lambda])\ed\nu_\lambda([v_\lambda]).
\end{equation}
For any $x=\frac{1}{m}J_\lambda([v_\lambda])=\action{h}{x_0}$ and $g\in G$ we can estimate the integrand using \eqref{eq:GactionHWorbitchi} as
\begin{equation}
\begin{split}
\Tr\rho^{\otimes m}p_{\lambda,\pi^{\otimes m}}([v_\lambda])
 & = \Tr\pi^{\otimes m}(g)^*\rho^{\otimes m}\pi^{\otimes m}(g)\pi^{\otimes m}(g^{-1})p_{\lambda,\pi^{\otimes m}}([v_\lambda])\pi^{\otimes m}(g^{-1})^*  \\
 & = \Tr\pi^{\otimes m}(g)^*\rho^{\otimes m}\pi^{\otimes m}(g)\chi_{mx}(g^{-1})p_{\lambda,\pi^{\otimes m}}([\action{k(g^{-1}h)}{v_\lambda}])  \\
 & \le \chi_{mx}(g^{-1})\Tr\pi^{\otimes m}(g)^*\rho^{\otimes m}\pi^{\otimes m}(g)  \\
 & = \left(\chi_x(g^{-1})\Tr\pi(g)^*\rho\pi(g)\right)^m  \\
 & = e^{-m(-\ln\chi_x(g^{-1})-\ln\Tr\pi(g)^*\rho\pi(g))}.
\end{split}
\end{equation}
We take the infimum over $g$ and then bound from above by the supremum over $x\in F$:
\begin{equation}
\Tr\rho^{\otimes m}p_{\lambda,\pi^{\otimes m}}([v_\lambda])\le e^{-m\ratefunction{\rho}(x)}\le e^{-m\inf_{x\in F}\ratefunction{\rho}(x)}.
\end{equation}
Since $\nu_\lambda$ is a probability measure, this value is also an upper bound on each integral.

For the highest weights $\lambda$ appearing in the decomposition of $\mathcal{H}^{\otimes m}$ we can estimate the dimension as $\dim\mathcal{H}_\lambda\le(m+1)^{\dim\mathcal{H}(\dim\mathcal{H}-1)/2}$. This can be seen by first decomposing into $U(\mathcal{H})$-isotypic components (which are also $K$-invariant subspaces) and then into $K$-isotypic ones and using the dimension formula for $U(\mathcal{H})$-representations corresponding to partitions of $m$ into at most $\dim\mathcal{H}$ parts.

It remains to bound the number of isomorphism classes of $K$-representations appearing in $\mathcal{H}^{\otimes m}$. These are distinguished by their highest weights, so we get an upper bound by counting the total number of different weights. The weights of $\mathcal{H}^{\otimes m}$ are sums of $m$ weights from $\mathcal{H}$ (with multiplicity), therefore their number is upper bounded by $(m+1)^{\dim\mathcal{H}}$.

Combining these estimates we get
\begin{equation}
\begin{split}
\mu_m(F)
 & \le \sum_{\lambda}\dim(\mathcal{H}_\lambda)\int_{J_\lambda^{-1}(mF)}e^{-\inf_{x\in F}\ratefunction{\rho}(x)}\ed\nu_\lambda([v_\lambda])  \\
 & \le \sum_{\lambda}(m+1)^{\dim\mathcal{H}(\dim\mathcal{H}-1)/2}e^{-\inf_{x\in F}\ratefunction{\rho}(x)}  \\
 & \le (m+1)^{\dim\mathcal{H}}(m+1)^{\dim\mathcal{H}(\dim\mathcal{H}-1)/2}e^{-\inf_{x\in F}\ratefunction{\rho}(x)}  \\
 & = (m+1)^{\dim\mathcal{H}(\dim\mathcal{H}+1)/2}e^{-\inf_{x\in F}\ratefunction{\rho}(x)}.
\end{split}
\end{equation}
as claimed.
\end{proof}
In particular, Proposition~\ref{prop:LDPupper} implies part~\ref{it:LDPupper} of Theorem~\ref{thm:LDP}.

\subsection{Law of large numbers}\label{sec:lawoflargenumbers}

As an application of the upper bound (Proposition~\ref{prop:LDPupper}) we now show that the measures $\mu_m$ converge weakly to the Dirac measure located at the value of the moment map.

\begin{proposition}\label{prop:hasminimum}
Let $C\subseteq i\mathfrak{k}^*$ closed. Then $\ratefunction{\rho}$ has a minimum on $C$.
\end{proposition}
\begin{proof}
$\ratefunction{\rho}$ is lower semicontinuous by Corollary~\ref{prop:lowersemicontinuous} and infinite outside a compact set by Proposition~\ref{prop:precompactdomain}, therefore it has a minimum on $C$.
\end{proof}

\begin{corollary}\label{cor:ratezeroatmomentmapvalue}
$\ratefunction{\rho}(J(\rho))=0$.
\end{corollary}
\begin{proof}
Apply Proposition~\ref{prop:LDPupper} to the closed set $C=i\mathfrak{k}^*$, using that the infimum is attained (Proposition~\ref{prop:hasminimum}):
\begin{equation}
\begin{split}
0=\limsup_{m\to\infty}\frac{1}{m}\ln\mu(i\mathfrak{k}^*)\le-\min_{x\in C}\ratefunction{\rho}(x),
\end{split}
\end{equation}
therefore $\ratefunction{\rho}$ vanishes in at least one point. By Proposition~\ref{prop:ratefunctionnonnegative} this is only possible for $x=J(\rho)$.
\end{proof}

\begin{proof}[Proof of Theorem~\ref{thm:LLN}]
By Proposition~\ref{prop:hasminimum}, $\ratefunction{\rho}$ has a minimum on the closed set $i\mathfrak{k}^*\setminus O$, say at $x$. The condition $J(\rho)\in O$ implies $x\neq J(\rho)$, so $\ratefunction{\rho}(x)>0$ by Proposition~\ref{prop:ratefunctionnonnegative}. From the upper bound \eqref{eq:LDPupper} we get
\begin{equation}
\begin{split}
\liminf_{m\to\infty}\mu_m(O)
 & = \liminf_{m\to\infty}\left(1-\mu_m(i\mathfrak{k}^*\setminus O)\right)  \\
 & \ge 1-\limsup_{m\to\infty}(m+1)^{\dim\mathcal{H}(\dim\mathcal{H}+1)/2}e^{-m\ratefunction{\rho}(x)}
 = 1.
\end{split}
\end{equation}
\end{proof}

\subsection{Lower bound}\label{sec:lowerbound}

For the lower bound we employ a variant of the ``change of measure'' or ``exponential tilting'' technique \cite{cramer1938nouveau}. However, instead of multiplying the measures $\mu_m$ by a suitable function, we replace $\rho$ with (the normalised version of) an element in its $G$-orbit so that we retain the form \eqref{eq:muintegral} and thus we can use Theorem~\ref{thm:LLN}.

The following lemma translates the rate of exponential decay of the probability of an open set to the decay of the probability density on the rescaled integral orbits. This equivalent characterisation will ease the comparison of the original and the tilted measures.
\begin{lemma}\label{lem:limequalssup}
Let $O\subseteq i\mathfrak{k}^*$ open. Then the limit $\lim_{m\to\infty}\frac{1}{m}\ln\mu_m(O)$ exists and is equal to
\begin{equation}\label{eq:limequalssup}
L=\sup_{\substack{m,\lambda,h  \\  \action{h}{\frac{\lambda}{m}}\in O}}\frac{1}{m}\ln\Tr\rho^{\otimes m}p_{\lambda,\pi^{m}}(\action{h}{[v_\lambda]}),
\end{equation}
where $m\in\mathbb{N}_{>0}$, $\lambda$ can be any dominant integral weight, $h\in K$, and $[v_\lambda]$ denotes the highest weight ray in $P\mathcal{H}_\lambda$.
\end{lemma}
\begin{proof}
As in the proof of Proposition~\ref{prop:LDPupper} we have $\mu_m(O)\le(m+1)^{\dim\mathcal{H}(\dim\mathcal{H}+1)/2}L$, which implies
\begin{equation}
\limsup_{m\to\infty}\frac{1}{m}\mu_m(O)\le L.
\end{equation}

For the other direction, let $\epsilon>0$, $m_0\in\mathbb{N}_{>0}$, $h\in K$, and $\lambda_0$ a dominant integral weight such that $\frac{1}{m_0}J_{\lambda_0}(\action{h}{[v_{\lambda_0}]})\in O$ and $\frac{1}{m_0}\ln\Tr\rho^{\otimes m}p_{\lambda_0,\pi^{m_0}}(\action{h}{[v_{\lambda_0}]})\ge L-\epsilon$, i.e.
\begin{equation}
\Tr\rho^{\otimes m_0}p_{\lambda_0,\pi^{m_0}}(\action{h}{[v_{\lambda_0}]})\ge e^{m_0(L-\epsilon)}.
\end{equation}
By continuity, there is an open neighbourhood $U\subseteq K$ of $h$ such that for every $h'\in U$ the inequality $\Tr\rho^{\otimes m_0}p_{\lambda_0,\pi^{m_0}}(\action{h'}{[v_{\lambda_0}]})\ge e^{m_0(L-2\epsilon)}$ holds and such that
\begin{equation}
\overline{\frac{1}{m_0}J_{\lambda_0}(\action{U}{[v_{\lambda_0}]})}\subseteq O.
\end{equation}
Let $\lambda_1$ be a dominant integral weight such that $\rho$ has nonzero overlap with the corresponding isotypic projection of $\mathcal{H}$. The map  $h\mapsto\Tr\rho p_{\lambda_1,\pi}(\action{h}{[v_{\lambda_1}]})$ ($K\to\mathbb{R}$) is analytic and not identically zero, therefore it is nonzero on a dense open set. Let $U'\subseteq U$ open such that $\Tr\rho p_{\lambda_1,\pi}(\action{h}{[v_{\lambda_1}]})\ge\delta>0$ for all $h\in U'$.

For $m\in\mathbb{N}$ let $q=\lfloor\frac{m}{m_0}\rfloor$ and $r=m-qm_0$. The tensor product of highest weight vectors is also a highest weight vector for the product representation and the weights are added. For $h\in U'$ and large $m$ we have $\action{h}{(q\lambda_0+r\lambda_1)}\in mO$, therefore
\begin{equation}
\begin{split}
\mu_m(O)
 & \ge \dim(\mathcal{H}_{q\lambda_0+r\lambda_1})\int_{U'}\Tr\rho^{\otimes m}p_{q\lambda_0+r\lambda_1,\pi^{\otimes m}}(\action{h}{[v_{\lambda_0}^{\otimes q}\otimes v_{\lambda_1}^{\otimes r}]})\ed\nu(h)  \\
 & \ge \int_{U'}\left(\Tr\rho^{\otimes m_0}p_{\lambda_0,\pi^{\otimes m_0}}(\action{h}{[v_{\lambda_0}]})\right)^q\left(\Tr\rho p_{\lambda_0,\pi}(\action{h}{[v_{\lambda_1}]})\right)^r\ed\nu(h)  \\
 & \ge e^{qm_0(L-2\epsilon)}\delta^r\nu(U')  \\
 & \ge e^{m(L-2\epsilon)}\delta^{m_0}\nu(U'),
\end{split}
\end{equation}
where $\nu$ is the Haar probability measure on $K$. This implies
\begin{equation}
\liminf_{m\to\infty}\frac{1}{m}\ln\mu_m(O)\ge\liminf_{m\to\infty}\frac{1}{m}\left(m(L-2\epsilon)+m_0\ln\delta+\ln\nu(U')\right)=L-2\epsilon.
\end{equation}
This is true for every $\epsilon>0$, therefore $L$ is also a lower bound.
\end{proof}

The next proposition introduces the tilted measures and compares them with the original one on open sets. We remark that, by the symmetry between $\rho$ and $\rho'$, a similar inequality holds in the reverse direction.
\begin{proposition}\label{prop:tilting}
Let $g\in G$, $\rho'=\frac{\pi(g)^*\rho\pi(g)}{\Tr\pi(g)^*\rho\pi(g)}$ and $O'\subseteq i\mathfrak{k}$ open. Consider the measures $\mu'_m:A\mapsto\Tr(\rho')^{\otimes m}E_{\mathcal{H}^{\otimes m}}(mA)$ (i.e. constructed as in \eqref{eq:muintegral} but with $\rho'$ instead of $\rho$). We have the inequality
\begin{equation}
\lim_{m\to\infty}\frac{1}{m}\ln\mu_m(O')\ge\lim_{m\to\infty}\frac{1}{m}\ln\mu'_m(\action{g^{-1}}{O'})-\sup_{y\in\action{g^{-1}}{O'}}\left[\ln\chi_y(g)-\ln\Tr\pi(g)^*\rho\pi(g)\right]
\end{equation}
\end{proposition}
\begin{proof}
We use Lemma~\ref{lem:limequalssup} for both $\mu'_m$ and $\mu_m$:
\begin{equation}
\begin{split}
\lim_{m\to\infty}\frac{1}{m}\ln\mu'_m(\action{g^{-1}}{O'})
 & = \sup_{\substack{m,\lambda,h  \\  \action{h}{\frac{\lambda}{m}}\in \action{g^{-1}}{O'}}}\frac{1}{m}\ln\Tr{\rho'}^{\otimes m}p_{\lambda,\pi^{\otimes m}}(\action{h}{[v_\lambda]})  \\
 & = \sup_{\substack{m,\lambda,h  \\  \action{h}{\frac{\lambda}{m}}\in \action{g^{-1}}{O'}}}\frac{1}{m}\ln\Tr\left(\frac{\pi(g)^*\rho\pi(g)}{\Tr\pi(g)^*\rho\pi(g)}\right)^{\otimes m}p_{\lambda,\pi^{\otimes m}}(\action{h}{[v_\lambda]})  \\
 & = \sup_{\substack{m,\lambda,h  \\  \action{h}{\frac{\lambda}{m}}\in \action{g^{-1}}{O'}}}\frac{1}{m}\ln\Tr\rho^{\otimes m}\pi(g)^{\otimes m}p_{\lambda,\pi^{\otimes m}}(\action{h}{[v_\lambda]})\pi(g^*)^{\otimes m}  \\  &\qquad-\ln\Tr\pi(g)^*\rho\pi(g)  \\
 & = \sup_{\substack{m,\lambda,h  \\  \action{h}{\frac{\lambda}{m}}\in \action{g^{-1}}{O'}}}\frac{1}{m}\ln\Tr\rho^{\otimes m}p_{\lambda,\pi^{\otimes m}}(\action{k(gh)}{[v_\lambda]})  \\  &\qquad+\ln\chi_{\frac{\lambda}{m}}(gh)-\ln\Tr\pi(g)^*\rho\pi(g)  \\
 & = \sup_{\substack{m,\lambda,h  \\  \action{h}{\frac{\lambda}{m}}\in \action{g^{-1}}{O'}}}\frac{1}{m}\ln\Tr\rho^{\otimes m}p_{\lambda,\pi^{\otimes m}}(\action{k(gh)}{[v_\lambda]})  \\  &\qquad+\ln\chi_{\action{h}{\frac{\lambda}{m}}}(g)-\ln\Tr\pi(g)^*\rho\pi(g)  \\
 & \le \sup_{\substack{m,\lambda,h  \\  \action{h}{\frac{\lambda}{m}}\in \action{g^{-1}}{O'}}}\frac{1}{m}\ln\Tr\rho^{\otimes m}p_{\lambda,\pi^{\otimes m}}(\action{k(gh)}{[v_\lambda]})  \\  &\qquad+\sup_{\substack{m,\lambda,h  \\  \action{h}{\frac{\lambda}{m}}\in \action{g^{-1}}{O'}}}\ln\chi_{\action{h}{\frac{\lambda}{m}}}(g)-\ln\Tr\pi(g)^*\rho\pi(g)  \\
 & = \lim_{m\to\infty}\frac{1}{m}\ln\mu_m(O')+\sup_{y\in \action{g^{-1}}{O'}}\ln\chi_y(g)-\ln\Tr\pi(g)^*\rho\pi(g).
\end{split}
\end{equation}
The fourth equality uses \eqref{eq:GactionHWorbitchi}, while the last step uses that $\action{h}{\frac{\lambda}{m}}\in \action{g^{-1}}{O'}$ if and only if $\action{k(gh)}{\frac{\lambda}{m}}\in O'$.
\end{proof}

\begin{proposition}\label{prop:specialx}
Let $g\in G$ and $x=\action{g}{J\left(\frac{\pi(g)^*\rho\pi(g)}{\Tr\pi(g)^*\rho\pi(g)}\right)}$. Then for every open set $O\subseteq i\mathfrak{k}^*$ such that $x\in O$ the inequality
\begin{equation}
\lim_{m\to\infty}\frac{1}{m}\ln\mu_m(O)\ge-\ratefunction{\rho}(x)
\end{equation}
holds. In addition,
\begin{equation}
\ratefunction{\rho}(x)=-\ln\chi_{x}(g^{-1})-\ln\Tr\pi(g)^*\rho\pi(g)<\infty.
\end{equation}
\end{proposition}
\begin{proof}
Let $\rho'=\frac{\pi(g)^*\rho\pi(g)}{\Tr\pi(g)^*\rho\pi(g)}$ and $x'=J(\rho')$, and consider the measures $\mu'_m$ as in Proposition~\ref{prop:tilting}. For any open set $O'$ such that $x=\action{g}{x'}\in O'\subseteq O$ we have
\begin{equation}
\begin{split}
\lim_{m\to\infty}\frac{1}{m}\ln\mu_m(O)
 & \ge \lim_{m\to\infty}\frac{1}{m}\ln\mu_m(O')  \\
 & \ge \lim_{m\to\infty}\frac{1}{m}\ln\mu'_m(\action{g^{-1}}{O'})-\sup_{y\in \action{g^{-1}}{O'}}\left[\ln\chi_y(g)-\ln\Tr\pi(g)^*\rho\pi(g)\right]  \\
 & = -\sup_{y\in \action{g^{-1}}{O'}}\left[\ln\chi_y(g)-\ln\Tr\pi(g)^*\rho\pi(g)\right]
\end{split}
\end{equation}
The first inequality follows from $O'\subseteq O$, the second inequality uses Proposition~\ref{prop:tilting}, and the equality is true by the law of large numbers (Theorem~\ref{thm:LLN}) using that $x'=J(\rho')\in \action{g^{-1}}{O'}$. Since $G$ acts on $i\mathfrak{k}^*$ by homeomorphisms, $\action{g^{-1}}{O'}$ can be made arbitrarily small by choosing $O'$ small. $y\mapsto\chi_y(g)$ is continuous, therefore
\begin{equation}\label{eq:LDPlowerspecialx}
\begin{split}
\lim_{m\to\infty}\frac{1}{m}\ln\mu_m(O)
 &  \ge -\left[\ln\chi_{x'}(g)-\ln\Tr\pi(g)^*\rho\pi(g)\right]  \\
 & =-\left[-\ln\chi_x(g^{-1})-\ln\Tr\pi(g)^*\rho\pi(g)\right].
\end{split}
\end{equation}

Apply this bound to the open ball $\ball{\delta}{x}$ instead of $O$ and let $\delta\to 0$:
\begin{equation}
\begin{split}
\ratefunction{\rho}(x)
 & \le \lim_{\delta\to 0}\inf_{y\in\ball{\delta}{x}}\ratefunction{\rho}(y)  \\
 & \le \lim_{\delta\to 0}-\limsup_{m\to\infty}\frac{1}{m}\ln\mu_m(\ball{\delta}{x})  \\
 & \le -\ln\chi_{x}(g^{-1})-\ln\Tr\pi(g)^*\rho\pi(g)  \\
 & \le \ratefunction{\rho}(x).
\end{split}
\end{equation}
Here the first inequality is lower semicontinuity (Corollary~\ref{prop:lowersemicontinuous}), the second inequality is the large deviation upper bound (Proposition~\ref{prop:LDPupper}), the third inequality follows from \eqref{eq:LDPlowerspecialx} and the last inequality holds by definition (Definition~\ref{def:ratefunctiondef}). This means that we have equality everywhere.
\end{proof}

Consider the set
\begin{equation}
\mathcal{M}_\rho=\setbuild{\action{g}{J\left(\frac{\pi(g)^*\rho\pi(g)}{\Tr\pi(g)^*\rho\pi(g)}\right)}}{g\in G}.
\end{equation}
The statements of Proposition~\ref{prop:specialx} can be rephrased as follows: for any open subset $O\subseteq i\mathfrak{k}^*$ the lower bound
\begin{equation}
\lim_{m\to\infty}\frac{1}{m}\ln\mu_m(O)\ge-\inf_{x\in O\cap\mathcal{M}_\rho}\ratefunction{\rho}(x)
\end{equation}
holds and $\mathcal{M}_\rho\subseteq\domain\ratefunction{\rho}$. We now show that $\mathcal{M}_\rho$ is a dense subset of $\domain\ratefunction{\rho}$.

\begin{proposition}\label{prop:specialxdense}
Let $x\in\domain\ratefunction{\rho}$. Then there is a sequence $g_1,g_2,\ldots\in G$ such that
\begin{equation}
\lim_{j\to\infty}\action{g_j}{J\left(\frac{\pi(g_j)^*\rho\pi(g_j)}{\Tr\pi(g_j)^*\rho\pi(g_j)}\right)}=x.
\end{equation}
\end{proposition}
\begin{proof}
Choose a sequence $g_1,g_2,\ldots\in G_x$ such that
\begin{equation}
\lim_{j\to\infty}\ln\chi_x(g_j)-\ln\Tr\pi(g_j)^*\rho\pi(g_j)=\ratefunction{\rho}(x),
\end{equation}
which is possible by Proposition~\ref{prop:ratefunctionexpressions}. Let $x=\action{h}{x_0}$ with $h\in K$ and $x_0\in\positivechamber$. By Lemma~\ref{lem:unipotentminlength} we can assume that $\Tr\pi(g_j)^*\rho\pi(g_j)\le\Tr\pi(g_jn)^*\rho\pi(g_jn)$ for every $j\in\mathbb{N}$ and $n\in hNh^{-1}$ (otherwise replace $g_j$ with a minimiser: this increases the sequence but the limit cannot increase since it was already equal to the supremum). This means that for $\nu\in \action{h}{\mathfrak{n}}$ we have
\begin{equation}
\begin{split}
0
 & = \left.\frac{\ed}{\ed s}\Tr\pi(g_j\exp s\nu)^*\rho\pi(g_j\exp s\nu)\right|_{s=0}  \\
 & = \Tr T_e\pi(\nu)^*\pi(g_j)^*\rho\pi(g_j)+\pi(g_j)^*\rho\pi(g_j)T_e\pi(\nu)  \\
 & = 2\Re\Tr\pi(g_j)^*\rho\pi(g_j)T_e\pi(\nu),
\end{split}
\end{equation}
and therefore $J(\frac{\pi(g_j)^*\rho\pi(g_j)}{\Tr\pi(g_j)^*\rho\pi(g_j)})\in \action{h}{i\mathfrak{t}^*}$.

Let $\rho_j=\frac{\pi(g_j)^*\rho\pi(g_j)}{\Tr\pi(g_j)^*\rho\pi(g_j)}$. By compactness of $\states(\mathcal{H})$ we may assume (after passing to a subsequence) that the sequence $\rho_1,\rho_2,\ldots$ converges to some state $\rho_\infty$. For every $g\in G_x$ we have
\begin{equation}
\begin{split}
\ratefunction{\rho}(x)
 & \ge \limsup_{j\to\infty}\ln\chi_x(g_jg)-\ln\Tr\pi(g_jg)^*\rho\pi(g_jg)  \\
 & = \limsup_{j\to\infty}\ln\chi_x(g_j)+\ln\chi_x(g)-\ln\Tr\pi(g)^*\rho_j\pi(g)-\ln\Tr\pi(g_j)^*\rho\pi(g_j)  \\
 & = \ratefunction{\rho}(x)+\limsup_{j\to\infty}\ln\chi_x(g)-\ln\Tr\pi(g)^*\rho_j\pi(g)  \\
 & = \ratefunction{\rho}(x)+\ln\chi_x(g)-\ln\Tr\pi(g)^*\rho_\infty\pi(g),
\end{split}
\end{equation}
therefore
\begin{equation}
0\ge\ln\chi_x(g)-\ln\Tr\pi(g)^*\rho_\infty\pi(g).
\end{equation}
Taking the supremum over $g\in G_x$ we obtain $0\ge\ratefunction{\rho_\infty}(x)$ which, by Proposition~\ref{prop:ratefunctionnonnegative} implies $x=J(\rho_\infty)$.

The sequence $x_j=J(\rho_j)$ converges to $x$ and satisfies $x_j\in\action{h}{i\mathfrak{t}}$, therefore for large $j$ we have $\action{g_j}{x_j}\in \action{K_x}{x_j}$ by Lemma~\ref{lem:KxorbitGxinvariant}. $K_x$ acts by isometries (for any $\Ad$-invariant inner product) and fixes $x$, therefore $\action{g_j}{x_j}\to x$.
\end{proof}

When $\ratefunction{\rho}$ is continuous on its domain, Propositions~\ref{prop:specialx} and~\ref{prop:specialxdense} imply \eqref{eq:LDPlower} for every open set. In general we have not been able to prove continuity, although we conjecture that it is indeed true. However, for states with full support Proposition~\ref{prop:invertiblecontinuity} is sufficiently strong to finish the proof as follows.
\begin{proof}[Proof of part~\ref{it:LDPlower} of Theorem~\ref{thm:LDP}]
Let $x\in O\cap\domain\ratefunction{\rho}$ and write $x=\action{h}{x_0}$ with $h\in K$ and $x_0\in\positivechamber$. Let $y\in\relativeinterior((\action{h}{\positivechamber})\cap\domain\ratefunction{\rho})$ be arbitrary. Then for $s\in[0,1)$ we have $sx+(1-s)y\in\relativeinterior((\action{h}{\positivechamber})\cap\domain\ratefunction{\rho})$, so by Proposition~\ref{prop:invertiblecontinuity} $\ratefunction{\rho}$ is continuous at these points. From Propositions~\ref{prop:specialx} and~\ref{prop:specialxdense} we conclude that
\begin{equation}
\lim_{m\to\infty}\frac{1}{m}\ln\mu_m(O)\ge-\ratefunction{\rho}(sx+(1-s)y).
\end{equation}
We take the limit $s\to 1$ and note that the function $s\mapsto\ratefunction{\rho}(sx+(1-s)y)$ is finite, convex and lower semicontinuous on $[0,1]$, therefore also continuous. Thus $-\ratefunction{\rho}(x)$ is a lower bound.
\end{proof}

\section*{Acknowledgements}

We acknowledge financial support from the European Research Council (ERC Grant Agreement no. 337603 and 818761) and VILLUM FONDEN via the QMATH Centre of Excellence (Grant no. 10059) and from Universidad de los Andes, Faculty of Sciences project INV-2017-51-1445 (AB). This research was supported by the J\'anos Bolyai Research Scholarship of the Hungarian Academy of Sciences and the National Research, Development and Innovation Fund of Hungary within the Quantum Technology National Excellence Program (Project Nr.~2017-1.2.1-NKP-2017-00001) and via the research grants K124152, KH129601 (PV).

\bibliography{refs}{}

\end{document}